\documentclass{article}

\usepackage[preprint]{staix_2026}      

\usepackage[utf8]{inputenc}
\usepackage[T1]{fontenc}
\usepackage{hyperref}
\usepackage{url}
\usepackage{booktabs}
\usepackage{amsfonts}
\usepackage{nicefrac}
\usepackage{microtype}
\usepackage{xcolor}
\usepackage{graphicx}
\usepackage{amsmath,amssymb,amsthm}
\usepackage{bm}
\usepackage{enumitem}

\usepackage{algorithm}
\usepackage{algpseudocode}

\theoremstyle{plain}
\newtheorem{theorem}{Theorem}[section]
\newtheorem{proposition}[theorem]{Proposition}

\theoremstyle{definition}

\newtheorem{assumption}[theorem]{Assumption}
\newtheorem{example}[theorem]{Example}
\theoremstyle{remark}

\usepackage[textsize=tiny]{todonotes}

\title{Global Average Treatment Effects for Individualized Randomization Experiments with Aggregate Data}

\author{Shuguang Yu$^{a}$\thanks{The first two authors have contributed equally to this paper.}, Ting Li$^{a}$, Yuchen Lu$^{a}$, Chengchun Shi$^{b}$, Fan Zhou$^{a}$, Zhichao Zou$^{c}$,\\ \textbf{Peng Zhen}$^{c}$, \textbf{and Hongtu Zhu}$^{d}$\thanks{Corresponding author. Email: htzhu@email.unc.edu}\\
$^a$Shanghai University of Finance and Economics, Shanghai, China\\
$^b$London School of Economics and Political Science, London, UK\\
$^c$DiDi China Ride Hailing Business Group, Beijing, China\\
$^d$University of North Carolina at Chapel Hill, North Carolina, USA}

\begin{document}

\maketitle

\begin{abstract}
Individualized randomized experiments are central to online platforms for optimizing personalized decisions in complex environments. In two-sided markets, however, standard treatment effect estimation is often invalid due to strong temporal and cross-unit interference, a challenge compounded when only aggregated data are available because of privacy or system constraints. To address these issues, we identify the Global Average Treatment Effect (GATE) using only group-level data from treatment and control groups. We first establish identification conditions based on aggregated observations, and then propose the Individualized Randomized Experiment Varying Coefficient Decision Process (IRE-VCDP) model, which accounts for interference through supply–demand dynamics. Building on this framework, we develop a complete procedure for estimation and statistical inference of the GATE, along with theoretical guarantees for the proposed test. Extensive simulations and real-world experiments using data from a leading ridesharing platform demonstrate the effectiveness of our approach.
\end{abstract}

\section{Introduction}

Individualized randomized experiments (IREs) assign treatments at the individual level over time, with treatment status held fixed within each experimental period. Policy evaluation under IREs is of particular importance in two-sided markets — such as ridesharing, food delivery, and freelance platforms — where interventions simultaneously affect both sides of the market and induce complex interactions and interference patterns across participants \citep{johari2022experimental, shi2023multiagent}. The primary objective is to introduce an intervention to a subset of the market and use the observed outcomes to estimate its effect if deployed at full scale. We refer to this estimand as the Global Average Treatment Effect (GATE, formally defined in \eqref{eq:GATE_def_individual}). Accurate estimation of GATE is central to platform decisions about whether to roll out a new policy to the entire market.

Our study is motivated by a concrete ridesharing marketplace setting, in which a city-scale online experiment is conducted to evaluate a new passenger subsidy policy. Passengers are randomly assigned at the outset to either a treatment group, receiving the new subsidy, or a control group, remaining under the existing policy, over a two- to four-week experimental period. In practice, individual-level data are often unavailable due to privacy constraints. Instead, only aggregated group-level metrics are observed in one-hour intervals, including Gross Merchandise Value (GMV), passenger demographic summaries, and key supply and demand variables such as order frequencies and driver online hours. The platform's goal is to quantify the causal impact of this policy on GMV under full deployment using only these aggregate observations.

 \textbf{Challenges.} Despite the widespread adoption of IREs in two-sided markets, reliable treatment effect estimation in this dynamic setting remains fundamentally difficult, for three reasons. The first challenge is cross-unit interference: interventions applied to one side of the market, for example, a passenger subsidy,  propagate through supply–demand interactions and affect outcomes on the other side, thereby violating the standard Stable Unit Treatment Value Assumption (SUTVA) \citep{imbens2015causal}. The second challenge is temporal interference: treatment effects evolve over time through dynamic market adjustments, so that an intervention at one time point can influence both contemporaneous and future outcomes \citep{wang2023robust, li2024evaluating, li2025testing}. Ignoring either form of interference leads to biased treatment effect estimates \citep{li2023optimal}. The third challenge is data availability: privacy and system constraints restrict access to individual-level records, leaving only aggregated group-level observations for analysis — a regime that remains largely underexplored in the causal inference literature. Taken together, these challenges render classical causal inference methods inadequate and motivate the development of new approaches tailored to two-sided markets.


\textbf{Contributions.} The primary goal of this paper is to study the GATE in individualized randomized experiments within two-sided markets using aggregated data, while explicitly accounting for both cross-unit and temporal interference. Our contributions are fourfold.
 
 (i) \textbf{Identification.} We establish identification conditions for the GATE using only group-level data from the treatment and control groups, together with market-level features. This result enables valid GATE estimation without access to individual-level records. To the best of our knowledge, this is the first work to provide such identification based solely on aggregated data.
 
 (ii) \textbf{Modeling framework.} We propose a two-layer modeling framework, the Varying Coefficient Decision Process (VCDP), that separates the outcome process from the underlying supply–demand dynamics. Interference is encoded through the supply–demand mechanism, reflecting the intrinsic structure of two-sided markets in which units interact via matching. The VCDP accommodates time-varying treatment effects and dynamic interference, with supply and demand variables serving as mediators that jointly capture cross-unit and temporal interference.

(iii) \textbf{Estimation.} We show that the GATE can be expressed as an explicit function of the VCDP model coefficients, and we develop a two-step estimation procedure that combines least squares estimation with kernel smoothing to improve efficiency and stability under weak signals and limited sample sizes. The GATE is then recovered via plug-in estimation from the fitted coefficients.

(iv) \textbf{Inference and empirical validation.} We develop a bootstrap-based testing procedure to assess whether deploying the new treatment across the full market improves platform revenue relative to the baseline policy, and we establish its consistency using Gaussian approximation theory \citep{chernozhukov2013gaussian}. We further validate the practical effectiveness of our approach through an empirical study on data from a leading ridesourcing platform, demonstrating improved performance relative to classical testing methods.

\subsection{Related Work}

\textbf{Causal inference in A/B testing.} The classical rationale for A/B testing rests on the SUTVA \citep{rubin1980randomization}, which requires that the treatment assigned to one unit does not affect the outcomes of others. In two-sided marketplaces, however, interference is typically unavoidable: interventions on one side of the market propagate through supply–demand interactions to affect outcomes on the other side, rendering standard A/B analyses causally invalid \citep{halloran2016dependent, savje2021average}.

Within the switchback design literature, there is a rich body of work on policy evaluation under interference \citep{hu2022switchback, bojinov2023design, farias2023correcting, luo2024policy, li2024evaluating}. By contrast, individualized randomized experiments (IREs) have received considerably less attention. The most closely related work is \citet{liu2023unite}, who proposed UniTE to estimate individual-level conditional average treatment effects in two-sided markets via Robinson decomposition. While effective in principle, UniTE requires granular transaction-level data that are often severely sparse in practice — in a two-week experiment, a single passenger may complete only one ride. Practitioners constrained to aggregated group-level data are thus largely limited to simple mean differences or difference-in-differences estimators, which are well known to be biased in the presence of interference \citep{imbens2015causal}. Bridging this gap between data availability and methodological requirements is the central motivation of this paper.

\textbf{Varying coefficient models.} Varying coefficient models, introduced by \citet{hastie1993varying}, allow regression coefficients to vary smoothly as functions of an index variable such as time. A rich estimation and inference theory has since developed, spanning kernel and local polynomial methods \citep{hoover1998nonparametric, fan1999statistical}, spline-based approaches for longitudinal data \citep{chiang2001smoothing, huang2004polynomial}, and broad applications in spatial, epidemiological, and longitudinal settings \citep{sun2023semiparametric, zhang2025bayesian, sadovnichiy2022mathematical}; see \citet{park2015varying} for a comprehensive survey. More recently, varying coefficient models have been applied to online A/B testing, though existing work is confined to switchback designs \citep{luo2024policy, li2024evaluating}. To the best of our knowledge, the present paper is the first to extend this framework to individualized randomized experiments, where interference must be handled through the supply–demand structure rather than the experimental design itself.

\section{Problem Formulation}

\textbf{Data.} Consider an experiment conducted over $n$ consecutive days, each divided into $m$ time intervals. At the outset, $N = N_0 + N_1$ subjects are randomly assigned to either a treatment group of size $N_1$ or a control group of size $N_0$, with assignments held fixed throughout the experiment. For day $d \in \{1,\dots,n\}$ and interval $t \in \{1,\dots,m\}$, subjects in the treatment group receive the new policy ($A_t^d = 1$) and subjects in the control group receive the baseline policy ($A_t^d = 0$). Let $Y_d^i(t) \in \mathbb{R}$ denote the outcome for subject $i$ at interval $t$ on day $d$.

Due to privacy constraints, individual-level outcomes are not observed. Instead, only group-level aggregates are available for each day $d$ and interval $t$. Specifically, for each group $C_k$ ($k \in \{0,1\}$), we observe three types of summaries: the average outcome
$
    Y_{d,C_k}(t) = N_k^{-1} \sum_{i \in C_k} Y_d^i(t),
$
the average covariate vector $X_{d,C_k}(t) \in \mathbb{R}^p$, and the group-level market variables $S_{d,C_k}(t) \in \mathbb{R}^q$, which capture supply and demand conditions faced by each group. Observations are assumed to be independent across days, while temporal dependence within each day is explicitly permitted.

\textbf{Objective}. 
We adopt the potential outcomes framework \citep{rubin2005causal} to formalize the problem. Let \( \bar{a}_t^i = (a_1^i, \dots, a_t^i)^\top \in \{0,1\}^t \) denote the treatment history of subject \( i \) up to time \( t \), and let \( \bar{\bm{a}}_{N,t} = (\bar{a}_t^1, \dots, \bar{a}_t^N) \) denote the joint treatment history for all \( N \) subjects. Define \( S(t, \bar{\bm{a}}_{N,t-1}) \) and \( Y^i(t, \bar{\bm{a}}_{N,t}) \) as the counterfactual market state and outcome for subject \( i \), respectively, under treatment history \( \bar{\bm{a}}_{N,t} \).
Our primary objective is to quantify the Global Average Treatment Effect (GATE), defined as the difference in expected cumulative outcomes between deploying the new policy to the entire market and maintaining the baseline policy:
\begin{equation}
\label{eq:GATE_def_individual}
\begin{aligned}
\text{GATE}
= \frac{1}{N}
\mathbb{E} \{\sum_{t=1}^m \sum_{i=1}^N  Y^i \big( t, \bm{1}_{N,t} \big)\}
- \frac{1}{N}
\mathbb{E} \{\sum_{t=1}^m \sum_{i=1}^N  Y^i \big( t, \bm{0}_{N,t} \big)\},
\end{aligned}
\end{equation}
where \( \bm{1}_{N,t} \) (respectively, \( \bm{0}_{N,t} \)) denotes that all entries in \( \bar{\bm{a}}_{N,t} \) are equal to 1 (respectively, 0).

Our goal is to estimate the GATE using aggregated data and to test the hypothesis
\begin{equation}
\label{eq:hypothesis}
H_0: \text{GATE} \leq 0 
\quad \text{versus} \quad 
H_1: \text{GATE} > 0,
\end{equation}
which evaluates whether deploying the new policy leads to an improvement over the baseline in terms of overall market outcomes.



\section{Identification}

The SUTVA \citep{imbens2015causal} is violated 
in our setting for two reasons. First, a subject's outcome depends not only on their own current 
treatment, but also on their prior treatment history. Second, the treatment trajectories of other 
subjects affect individual outcomes through shared supply--demand dynamics, inducing both cross-unit 
and temporal interference. Without additional structure on the nature of this interference, 
identifying the treatment effect may be impossible \citep{basse2018limitations, liu2024cluster}. 
We therefore impose the following four assumptions to identify the GATE from the observed 
aggregate data.

\begin{assumption}[Dynamic local interference]
    \label{assump: ATRA}
    For each $t$, if $a_t^i = a_t^{i\prime}$ and the fraction of treated subjects coincides 
    under $\bar{\bm{a}}_{N,t}$ and $\bar{\bm{a}}^\prime_{N,t}$, then 
    $Y^i(t, \bar{\bm{a}}_{N,t}) = Y^i(t, \bar{\bm{a}}^\prime_{N,t})$ 
    and 
    $S(t+1, \bar{\bm{a}}_{N,t}) = S(t+1, \bar{\bm{a}}^\prime_{N,t})$.
\end{assumption}

\begin{assumption}[Consistency]
    \label{assump: CA}
    If the treatment policy satisfies $\bar{\bm{A}}_{N,t} = \bar{\bm{a}}_{N,t}$, then 
    $S(t) = S(t, \bar{\bm{a}}_{N,t-1})$ and $Y^i(t) = Y^i(t, \bar{\bm{a}}_{N,t})$.
\end{assumption}

\begin{assumption}[Sequential ignorability]
    \label{assump: SRA}
    The treatment $\bm{A}^i(t)$ assigned to subject $i$ at time $t$ is conditionally independent 
    of all potential variables given the observed data history.
\end{assumption}

\begin{assumption}[Positivity]
    \label{assump: PA}
    For any $t \geq 1$, the probability of $\{A^i(t) = 1\}$ given the current state is strictly 
    bounded away from zero and one.
\end{assumption}

We now discuss the role and justification of each assumption. Assumption~\ref{assump: ATRA} 
restricts the interference structure: it requires that a subject's outcome and the subsequent 
market state depend on other subjects' treatments only through the aggregate fraction of treated 
subjects, not through their individual identities or specific assignments. This mean-field 
interference structure is well established in the exposure mapping literature 
\citep{yang2018mean, li2022random, shi2023multiagent} and generalizes the static local 
interference assumption of \citet{johari2022experimental} and \citet{masoero2026multiple} to a 
dynamic setting. It is particularly natural in two-sided markets, where individual outcomes are 
driven by market-level supply and demand conditions rather than by the identity of any particular 
treated peer.

Assumption~\ref{assump: CA} is a standard consistency condition: it requires that potential 
outcomes and states, under the observed treatment history, coincide with the actually observed 
outcomes and states. Assumption~\ref{assump: SRA} is a sequential ignorability condition 
requiring that, given the accumulated data history, the treatment at each period is conditionally 
independent of all remaining potential variables; this is a standard assumption in the dynamic 
treatment literature \citep{wang2018quantile}. Assumption~\ref{assump: PA} ensures that each 
subject has a nonzero probability of receiving either treatment at every period, given the current 
state. Under the randomized design considered here, where subjects are assigned to treatment or 
control at the outset with fixed positive probability, Assumption~\ref{assump: PA} is 
automatically satisfied.

\begin{proposition}
    \label{prop: identification}
    Suppose Assumptions~\ref{assump: ATRA}--\ref{assump: PA} hold, and that there exists a 
    function $R(\cdot)$ such that
    \begin{align}
        \mathbb{E} \Big( Y^i (t, \bar{\bm{a}}_{N,t}) \,\Big|\,
        \bar{\bm{a}}_{N,t}, \{ S(j, \bar{\bm{a}}_{N,j-1}), f(j) \}_{j \leq t} \Big)
        = R \Big(
            a_t^i,\,
            \{ S(j, \bar{\bm{a}}_{N,j-1}), f(j) \}_{j \leq t}
        \Big),
        \label{eq: R assumption}
    \end{align}
    where $f(j)$ denotes the fraction of treated subjects at time $j$. Then the following 
    results hold.

    \noindent\textbf{(i) Identification of $R(\cdot)$ from observed data.} 
    Under the full-treatment policy $\bar{\bm{a}}_{N,t} = \bm{1}_{N,t}$ and the 
    full-control policy $\bar{\bm{a}}_{N,t} = \bm{0}_{N,t}$, the function $R(\cdot)$ 
    can be learned from the observed data as
    \begin{align}
        R \Big(
            a,\, \{ S(j), f(j) \}_{j \leq t}
        \Big)
        = \mathbb{E} \Big(
            N_a^{-1} \sum_{i \in C_a} Y^i(t)
            \,\Big|\,
            \{ S(j), f(j) \}_{j \leq t}
        \Big),
        \label{eq: definition of R}
    \end{align}
    where $N_a$ denotes the number of subjects with $a_t^i = a$, for $a \in \{0, 1\}$.

    \noindent\textbf{(ii) Identification of GATE from observed data.} 
    The GATE defined in~\eqref{eq:GATE_def_individual} is identifiable from the observed 
    summary data. Specifically,
    \begin{eqnarray} \label{eq: result of idetification} 
        && \mathbb{E} \Big( Y^i (t, \bar{\bm{a}}_{N,t}) \Big)
        = \mathbb{E} \Big[
            R \Big(
                a,\, \{ S(j, \bar{\bm{a}}_{N,j-1}), f(j) \}_{j \leq t}
            \Big)
        \Big]          \\
        &=& \mathbb{E} \Big[
            \mathbb{E} \Big\{
                R \Big(
                    a,\, \{ S(j), f(j) \}_{j \leq t}
                \Big)
                \,\Big|\,
                \bar{\bm{A}}_{N,t} = \bar{\bm{a}}_{N,t},\,
                \{ S(j), f(j) \}_{j \leq t},\,
                \{ Y(j) \}_{j < t}
            \Big\}
        \Big]. \nonumber
    \end{eqnarray}
\end{proposition}

Proposition \ref{prop: identification} shows that the GATE can be presented as a function of the observed summary data, making it identifiable without access to the individual-level data. When $\bar{\bm{a}}_{N,t} = \bm{1}_{N,t}$, the corresponding ratio is $f(t) = 1$; when $\bar{\bm{a}}_{N,t} = \bm{0}_{N,t}$, we have $f(t) = 0$.


\section{Models for the outcome and state variables}
\label{sec:models}

We introduce the IRE-VCDP model to jointly characterize outcome and state dynamics in 
two-sided markets. The state variables explicitly encode the platform's demand--supply system. 
For each group $C \in \{C_0, C_1\}$, let $D_{d,C}(t)$ denote group-specific demand and 
$S_d(t)$ denote platform-wide supply, which is shared across groups since drivers serve all 
users regardless of group membership. We define the group-level state vector as 
$S_{d,C}(t+1) = (D_{d,C}(t+1), S_d(t+1))^\top$.

The outcome and state evolution are modeled as
\begin{eqnarray*}
    Y_{d,C}(t) &=& g_C\big(t,\, X_{d,C}(t),\, S_{d,C}(t)\big) + \epsilon_{d,C}(t), \\
    S_{d,C}(t+1) &=& G_C\big(t,\, \widetilde{X}_{d,C}(t),\, S_{d,C}(t),\, f_d(t)\big) 
    + E_{d,C}(t+1),
\end{eqnarray*}
where $g_C(\cdot)$ and $G_C(\cdot)$ are unknown regression functions. The outcome is driven 
by group characteristics $X_{d,C}(t)$ and the current demand--supply state $S_{d,C}(t)$, 
while the state transition depends additionally on the fraction of treated subjects $f_d(t)$, 
which serves as the channel through which the treatment propagates across groups. The error 
terms $\epsilon_{d,C}(t)$ and $E_{d,C}(t+1)$ are mutually independent; however, we allow 
correlation between $\epsilon_{d,C_0}(t)$ and $\epsilon_{d,C_1}(t)$ to accommodate shared 
latent shocks at the platform level. The inclusion of covariates $X_{d,C}(t)$ and 
$\widetilde{X}_{d,C}(t)$ improves estimation efficiency \citep{su2021model}.

To enable tractable estimation and interpretation, we adopt the following linear approximations:
{\small \begin{eqnarray}
    \label{eq:model_Y}
   && Y_{d,C}(t) = \alpha_{0,C}(t) + \alpha_{1,C}^\top(t)\, X_{d,C}(t) 
    + \alpha_{2,C}^\top(t)\, S_{d,C}(t) + \epsilon_{d,C}(t), \\
    \label{eq:model_state_within_group_all}
    && S_{d,C}(t+1) = \gamma_{0,C}(t) + f_d(t)\,\gamma_{1,C}(t) 
    + \Phi_{0,C}(t)\,\widetilde{X}_{d,C}(t) 
    + \Phi_{1,C}(t)\,S_{d,C}(t) + E_{d,C}(t+1).
\end{eqnarray} 
} This specification is motivated by both structural and empirical considerations. Structurally, 
it reflects the platform's operational asymmetry: demand is group-specific, arising from 
user-level randomization, while supply is shared across groups. This asymmetry is the mechanism 
through which cross-group interference operates --- a demand shock in one group affects service 
availability for the other through the common supply pool, and consequently affects outcomes. 
Empirically, linear varying coefficient models of this form have been shown to perform well in 
ridesharing systems \citep{luo2024policy, li2024evaluating}, providing a parsimonious yet 
flexible representation of dynamic demand--supply interactions.

A key feature of two-sided markets is that the imbalance between supply and demand is a primary 
determinant of outcomes \citep{chen2024real}. When supply is sufficient, demand and GMV exhibit 
an approximately linear relationship. Under excess demand, however, a fraction of ride requests 
cannot be fulfilled, and this congestion breaks down the linear relationship. To capture this 
nonlinearity, we incorporate the supply--demand gap
\begin{eqnarray}
    \label{eq:demand_supply_gap}
    G_d(t) = \max\{\mathrm{Demand}_d(t) - \mathrm{Supply}_d(t),\, 0\}
\end{eqnarray}
as a component of $X_{d,C}(t)$, where $\mathrm{Demand}_d(t)$ is the number of ride requests 
and $\mathrm{Supply}_d(t)$ is the number of rides that can be fulfilled given total available 
driver hours, both in city $d$ at time $t$. The gap $G_d(t)$ captures congestion effects that 
are particularly pronounced during peak periods and play a crucial role in explaining GMV 
dynamics over the experimental horizon, as confirmed by our empirical analysis.

Based on models \eqref{eq:model_Y}--\eqref{eq:model_state_within_group_all}, the GATE in \eqref{eq:GATE_def_individual} admits a closed-form representation as a function of the model coefficients, which greatly facilitates estimation and inference.
The coefficients in \eqref{eq:model_Y}--\eqref{eq:model_state_within_group_all} can be consistently estimated using the observed aggregate data. The GATE estimator is then obtained by plugging these estimates into \eqref{eq: result of idetification}, with $f(t)=1$ corresponding to full-scale treatment and $f(t)=0$ corresponding to the control.

\begin{proposition}
	\label{prop: ATE expression}
Suppose the conditions in Proposition \ref{prop: identification} holds, and the potential outcomes admits the form of models in \eqref{eq:model_Y}-\eqref{eq:model_state_within_group_all}, then we have
{\small
\begin{align}
\label{eq:GATE_coe_expression}
	&\text{GATE} =
\sum_{t=1}^m 
(  \alpha_{0,C_1} (t) -  \alpha_{0,C_0}(t) ) 
+\sum_{t=1}^m \Big[ \bar X_N ( \alpha_{1,C_1} (t)  -  \alpha_{1,C_0} (t) )  \Big]+ \\ 
&\sum_{t=1}^m \Big[
\alpha_{2,C_1} (t)^\top   \Big(\prod_{l=1}^{t-1} \Phi_{1,C_1}(l) \mathbb{E}(S_{C_1}(1))+ \sum_{k=1}^{t-1} \{ \prod_{l=k+1}^{t-1} \Phi_{1,C_1}(l)
\cdot[ \gamma_{0, C_1}(k)+\gamma_{1,C_1} (k)+\Phi_{0, C_1} (t-1) \bar{\tilde{X}}_N(t) ] \}\Big)
 \nonumber \\
&-
\alpha_{2,C_0} (t)^\top \Big(\prod_{l=1}^{t-1} \Phi_{1, C_0}(l) \mathbb{E}(S_{C_0}(1))+ \sum_{k=1}^{t-1} \{ \prod_{l=k+1}^{t-1} \Phi_{1, C_0}(l) 
\cdot[\gamma_{0, C_0}(k)+\Phi_{0, C_0} (t-1) \bar{\tilde{X}}_N(t) ]\}\Big)
\Big],  \nonumber
\end{align}
}
where $\bar X_N=\mathbb{E}( N^{-1} \sum_{i=1}^N X^i(t)  ) $ and $\bar{\tilde{X}}_N(t) = \mathbb{E} ( N^{-1}  \sum_{i=1}^N \tilde{X}^i(t) )$, 
$\mathbb{E}(S(1))$ denotes the demand and supply of all subjects, and $\prod_{l=k+1}^{t-1} \Phi_{1,C_0}(l)=\prod_{l=k+1}^{t-1} \Phi_{1,C_1}(l) =1$, if $t-1 < k+1$.
\end{proposition}

From Proposition~\ref{prop: ATE expression}, the GATE can be estimated by plugging the estimated coefficients from models \eqref{eq:model_Y}--\eqref{eq:model_state_within_group_all} into the closed-form expression. The resulting GATE admits a natural decomposition into three components: (i) the direct effect of treatment on the outcome at the same time point; (ii) the interaction between treatment and predictive covariates; and (iii) the interference effects across subjects and over time.

If the predictive covariates are centered to have zero mean, the expression of the \text{GATE} in Proposition~\ref{prop: ATE expression} further simplifies. In practice, this can be achieved by centering the covariates across all subjects (rather than within each group). While such centering simplifies the expression, the inclusion of these covariates also improves estimation efficiency by reducing variance.


To better illustrate the role of interference, we consider a naive treatment effect defined as the difference in mean outcomes between the two groups:
$
\tau = \mathbb{E}\Big( \sum_{t=1}^m Y_{C_1}(t) \Big) 
- \mathbb{E}\Big( \sum_{t=1}^m Y_{C_0}(t) \Big).
$
This estimand ignores interference across subjects.
Under models \eqref{eq:model_Y}--\eqref{eq:model_state_within_group_all}, $\tau$ can also be expressed in terms of the model coefficients, capturing both the direct treatment effect and its interactions with predictive and state variables. It can be viewed as a dynamic extension of the regression-adjusted estimator of \citet{li2020rerandomization}. 
However, $\tau$ fails to account for temporal and cross-subject interference. Consequently, it coincides with the \text{GATE} only when the expected demand and supply processes are identical under full-scale treatment and under control—an assumption that is typically violated in practice.
See Appendix~\ref{app:Formulation of Straightforward Treatment Effect} for detailed derivations and discussion.


\section{Estimation and Inference}

This section presents the estimation procedure and statistical inference framework for the
hypothesis~\eqref{eq:hypothesis}, based on
models~\eqref{eq:model_Y}--\eqref{eq:model_state_within_group_all}.

{\bf Estimation}. 
Estimation of the GATE proceeds in two steps. In the first step, the time-varying coefficients
in models~\eqref{eq:model_Y}--\eqref{eq:model_state_within_group_all} are estimated by ordinary
least squares at each time point. Specifically, for $t = 1, \ldots, m$ (and $t = 1, \ldots, m-1$
for the state equation), define
\begin{align*}
    \widehat{\theta}_{C}(t)
        &= \arg\min_{\theta_{C}(t)}
           \sum_{d=1}^n
           \Bigl( Y_{d,C}(t) - Z_{d,C}(t)^\top \theta_{C}(t) \Bigr)^2, \\
    \widehat{\Theta}_C^{(\nu)}(t)
        &= \arg\min_{\Theta^{(\nu)}(t)}
           \sum_{d=1}^n
           \Bigl( S_{d,C}^{(\nu)}(t+1) - \widetilde{Z}_{d,C}(t)^\top \Theta_C^{(\nu)}(t) \Bigr)^2,
\end{align*}
where $C \in \{C_0, C_1\}$,
$Z_{d,C}(t) = [1,\, X_{d,C}(t)^\top,\, S_{d,C}(t)^\top]^\top$,
$S_{d,C}^{(\nu)}(t+1)$ denotes the $\nu$-th component of $S_{d,C}(t+1)$, and
$\widetilde{Z}_{d,C}(t) = [1,\, f(t),\, \widetilde{X}_{d,C}(t)^\top,\, S_{d,C}(t)^\top]^\top$.

In the second step, kernel smoothing is applied to the initial estimates to reduce variance,
enforce temporal smoothness, and improve signal detection under weak effect
sizes~\citep{zhu2014spatially}. For a kernel function $K(\cdot)$ with bandwidth $h$, the smoothed
estimators are defined as
\begin{equation}
\label{eq:coe_est}
    \widetilde{\theta}_{C}(t)
        = \sum_{j=1}^{m} \omega_{j,h}(t)\,\widehat{\theta}_{C}(j)
    \qquad
    \mbox{and}  \qquad
    \widetilde{\Theta}^{(\nu)}_C(t)
        = \sum_{j=1}^{m-1} \omega_{j,h}(t)\,\widehat{\Theta}^{(\nu)}_C(j),
\end{equation}
where
$\omega_{j,h}(t) = K\!\left((j-t)/(mh)\right)\big/\sum_{k=1}^{m} K\!\left((k-t)/(mh)\right)$
are normalized kernel weights. We adopt the Gaussian kernel $K(t) = \exp(-t^2)$. The smoothed
estimators are weighted averages of the initial pointwise estimates, with weights decaying as a
function of temporal distance. This yields continuous, stable coefficient trajectories that are
particularly well-suited for settings in which the signal is weak or the underlying parameters
evolve slowly over time.

The GATE estimator $\widehat{\text{GATE}}$ is then obtained by substituting
$\widetilde{\theta}_{C}(t)$ and $\widetilde{\Theta}^{(\nu)}_C(t)$ into the closed-form expression
of Proposition~\ref{prop: ATE expression}, with the population means of the predictive covariates
and the initial state variable replaced by their empirical counterparts across all subjects
(see Appendix~\ref{app:Formulation of GATE Estimator}).

{\bf Inference}. 
To test the hypothesis~\eqref{eq:hypothesis}, we use the test statistic
$T = \widehat{\text{GATE}}$.
Under the null hypothesis, $T$ is expected to be non-positive or close to zero, so we reject
$H_0$ when $T$ exceeds a sufficiently large positive threshold. Deriving the exact limiting
distribution of $T$ is analytically intractable due to its complex functional dependence on the
estimated model parameters. We therefore approximate the null distribution of $T$ via a multiplier
bootstrap procedure adapted from~\citet{chernozhukov2013gaussian}.

Specifically, after estimating the outcome and state models, fitted values and residuals are computed within each group. The residuals are then perturbed using independent standard normal multipliers to generate bootstrap pseudo-data, thereby mimicking the variability of the original observations.
For each bootstrap sample, the full estimation procedure is repeated using the pseudo-data to obtain a bootstrap GATE estimate $\widehat{\text{GATE}}^{\,b} $. The bootstrap statistic is defined as 
$T^b = \widehat{\text{GATE}}^{\,b} - \widehat{\text{GATE}}.$
Repeating this procedure over $B$ replications provides an empirical approximation of the null limit distribution. Then we reject $H_0$ at significance level $\alpha$ if $T$ excedds the $1-\alpha$ quantile of $\{T_b\}_{b=1}^B$.

The complete estimation and inference procedure is summarized in
Algorithm~\ref{alg:estimation and testing} in Appendix \ref{sec:algorithm}. The bootstrap-based critical value enables valid
inference without requiring a closed-form asymptotic distribution.
The theoretical validity of Algorithm~\ref{alg:estimation and testing} is established below.

\begin{theorem}
\label{thm:bootstrap_consistency_QTE}
Suppose Assumptions~\ref{assump:kernel}--\ref{asmp:st1} in the supplement hold.
If $h = o(n^{-1/4})$, $m \asymp n^{c_2}$ for some $\tfrac{1}{2} < c_2 < \tfrac{3}{2}$,
and $mh \to \infty$ as $n \to \infty$, then
\begin{equation*}
    \sup_{z}
    \left|
        \mathbb{P}\!\left(T - \mathrm{GATE} \leq z\right)
        -
        \mathbb{P}\!\left(T^b - T \leq z \mid \mathrm{Data}\right)
    \right|
    \;\leq\;
    \widetilde{C}
    \!\left(
        \sqrt{n}\,h^2
        + \sqrt{n}\,m^{-1}
        + n^{-1/8}
    \right)
\end{equation*}
holds with probability approaching $1$, for some positive constant $\widetilde{C}$.
\end{theorem}

Theorem~\ref{thm:bootstrap_consistency_QTE} establishes that the bootstrap distribution of
$T^b - T$ consistently approximates the sampling distribution of $T - \mathrm{GATE}$, with an
explicit convergence rate. The bound comprises three terms: a smoothing bias $O(\sqrt{n}\,h^2)$
from kernel estimation, a discretization error $O(\sqrt{n}\,m^{-1})$ due to the finite number of
time intervals, and the Gaussian approximation error $O(n^{-1/8})$ inherited from the theory
of~\citet{chernozhukov2013gaussian}. The bandwidth conditions $h = o(n^{-1/4})$ and $mh \to \infty$
balance these three sources of error and ensure that the overall bound vanishes as $n \to \infty$.
The proof is given in Appendix~\ref{app:bootstrap_consistency_QTE}.

\section{Experiments}\label{sec:experiments}


In this section, we conduct real-data–based simulation studies to assess the finite-sample performance of the proposed test for \eqref{eq:hypothesis}. We then apply the method to real data from a leading ridesharing company. We compare our procedure against three alternatives: (i) a two-sample $t$-test, which is currently used within the company; (ii) a Difference-in-Differences (DiD) estimator, which compares pre-post outcome changes between treatment and control groups \citep{ashenfelter1978estimating,callaway2021difference}; and (iii) a Direct Effect (DE) estimator, which estimates the treatment effect by extending the regression adjustment approach in \cite{li2020rerandomization} to a dynamic setting. The outcome is the GMV, the state variables are the number of order requests and the drivers's total online time within each $t$. Detailed data-generating process and additional results are provided in Appendix~\ref{appendix:additional_experiment}.

\begin{example}[{\bf Real-data-based simulations within cities}]
\label{ex:simulation_single_city}

We construct a simulation environment based on real A/A experimental data from a ridesharing
platform across three cities of different sizes: a large city (5--10 million residents), a medium
city (1--5 million), and a small city (0.5--1 million). In an A/A experiment, identical policies
are applied to both groups, so the true GATE is zero by construction. Each day is divided into
$m = 24$ time intervals, and the number of experimental days is set to $n \in \{14, 28\}$. To
assess power, we introduce a tunable parameter $\eta \in [0,12]$ that scales the treatment effect,
with $\eta = 0$ corresponding to the null and larger values of $\eta$ yielding stronger effects.
All results are based on 1{,}000 simulation replications, each using 500 bootstrap samples.

Figure~\ref{fig:city_power} summarizes the results. The left panel shows that the bootstrap
distribution closely tracks the empirical null distribution of $\widehat{\text{GATE}}$ across
replications, confirming the accuracy of the bootstrap approximation established in
Theorem~\ref{thm:bootstrap_consistency_QTE}. 
The right panel reports the empirical rejection rates of the proposed IRE-VCDP-based test, the two-sample $t$-test, the DiD estimator, and the DE estimator. All methods control the type I error at the nominal 5\% level; however, the proposed test demonstrates superior power, with particularly notable gains for small effect sizes.
Figure~\ref{fig:city_pvalues} in the Appendix further confirms that the empirical $p$-value
distributions under the null are approximately uniform across all city sizes and observation
windows, indicating reliable Type~I error control.

\end{example}

\begin{figure}[htbp]
    \centering
    \includegraphics[width=0.44\columnwidth]{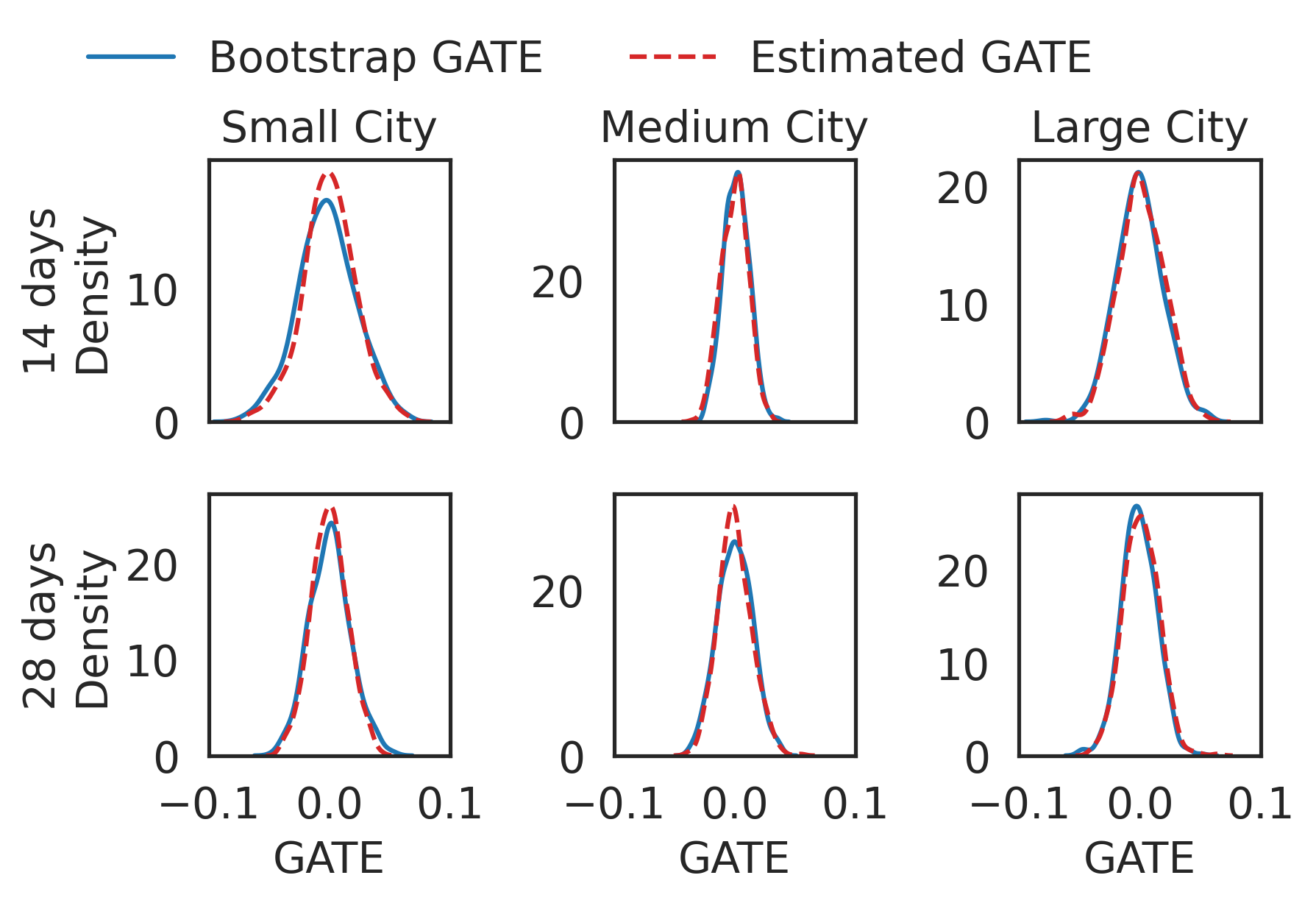}
    \includegraphics[width=0.48\columnwidth]{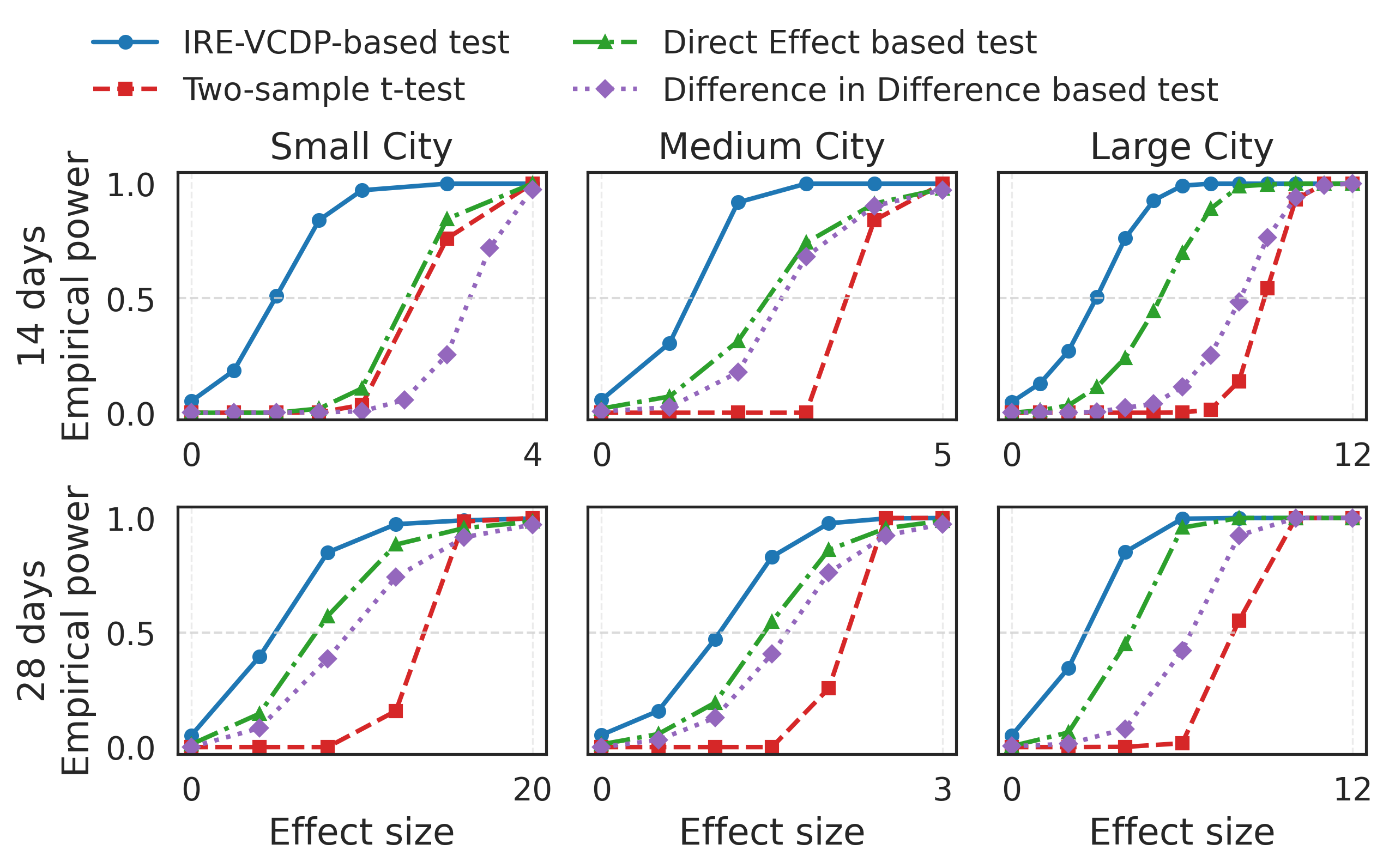}
    \caption{Left panel: empirical and bootstrap sampling distributions of the estimated GATE under null; Right panel: empirical rejection rates of the proposed IRE-VCDP-based test, the $t$-test, the DiD estimator, and the DE estimator.}
    \label{fig:city_power}
\end{figure}

\begin{example}[{\bf Real-data-based simulations across all cities}]
\label{ex:simulation_all_city}

Beyond individual cities, our framework also accommodates joint evaluation across multiple cities. In this example, we define the GATE as the overall difference in GMV across the three cities between the treatment and control groups. The findings are qualitatively similar to those obtained for individual cities: the bootstrap approximation remains accurate under the null and the proposed test maintains a clear power advantage over all three alternative methods. Detailed results are deferred to Appendix~\ref{appendix:additional_experiment}.

\end{example}

\begin{example}[{\bf Real data analysis}]
\label{ex:real_data}

We apply our method to a real-world A/B experiment conducted by Didi Chuxing from November~1 to
November~21, 2025 (21 days), in which users were randomly assigned to treatment and control groups
receiving different subsidy policies. The new policy was designed to increase GMV relative to the
incumbent one.

Before fitting the models, we examine the raw data. Figure~\ref{fig:cities_three_plots} in the
Appendix displays the temporal evolution of GMV, order requests, and drivers' total online time
across cities over the experimental period. GMV closely tracks order volume under normal conditions,
but visibly decouples on days with abnormally high demand, when requests exceed available supply
and cannot be fully converted into completed rides. This pattern motivates the inclusion of the
supply--demand imbalance term~\eqref{eq:demand_supply_gap} and confirms the suitability of
models~\eqref{eq:model_Y}--\eqref{eq:model_state_within_group_all} for this setting.

Figure~\ref{fig:city_140_res_plots} presents, for the medium city, the fitted values of GMV,
order requests, and drivers' total online time against their observed counterparts, together with
the corresponding residuals. Figures~\ref{fig:city_53_res_plots_large}
and~\ref{fig:city_310_res_plots_small} in the Appendix show analogous results for the large and small cities,
respectively. Across all cities, the proposed models achieve a close fit to both the outcome and
state variables, with residuals that fluctuate randomly and exhibit no systematic structure.


Table~\ref{tab:pvalues} in Appendix~\ref{appendix:additional_experiment} reports the $p$-values from all four methods for both A/A and A/B analyses across cities. The A/A results serve as a sanity check: all methods yield $p$-values above 5\%, confirming proper Type I error control. For the A/B analysis, the proposed test shows a meaningful drop in $p$-values from the A/A period, with the small city reaching 0.078 (significant at the 10\% level); no significant effects are detected for the medium or large cities. In contrast, the alternative methods yield similar $p$-values in both A/A and A/B periods, failing to distinguish between them, suggesting limited sensitivity to treatment effects transmitted through supply--demand dynamics.

\end{example}

\begin{figure}[htbp]
    \centering
    \includegraphics[width=0.9\columnwidth]{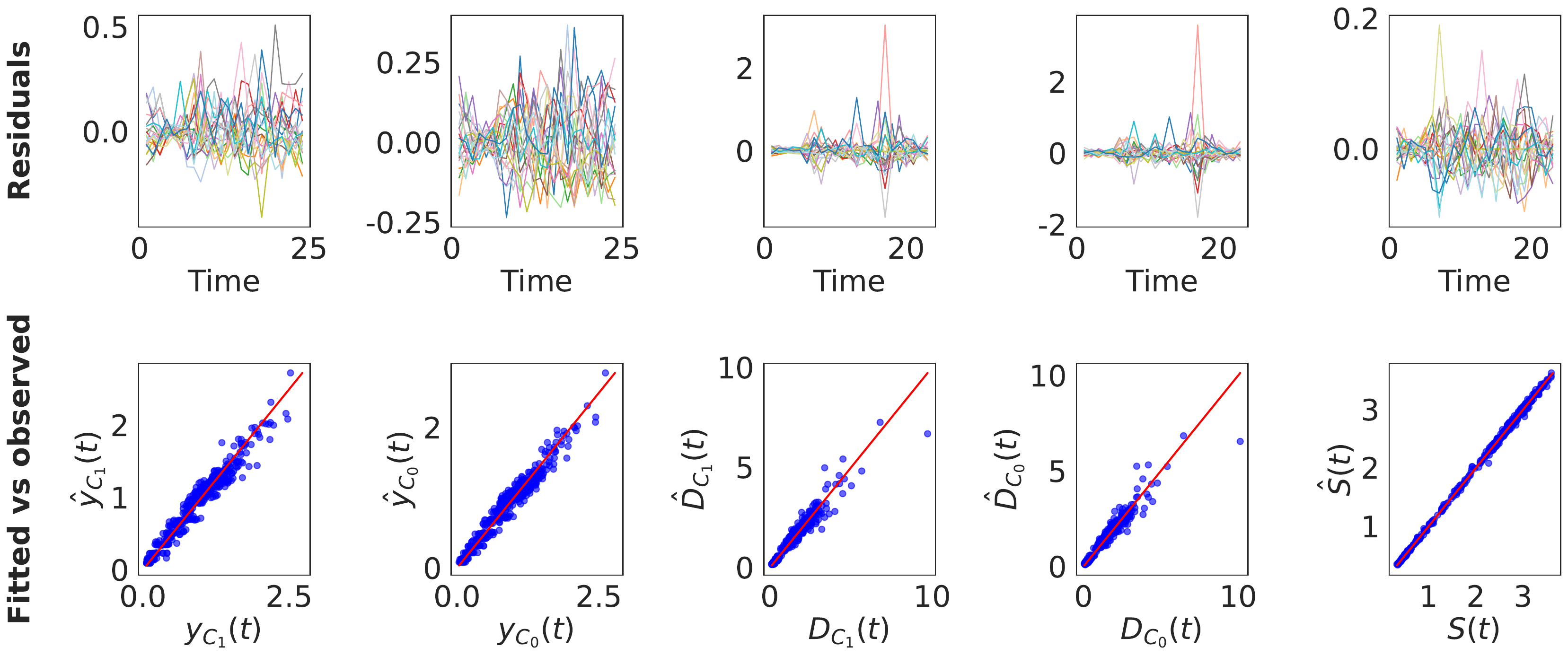}
    \caption{Medium city: residuals (top row) and fitted vs. observed values (bottom row). Columns (left to right): GMV (treatment), GMV (control), number of requests (treatment), number of requests (control), and drivers’ total online time.}
    \label{fig:city_140_res_plots}
\end{figure}

\bibliographystyle{plainnat}
\bibliography{example_paper}

@article{callaway2021difference,
  title={Difference-in-differences with multiple time periods},
  author={Callaway, Brantly and Sant’Anna, Pedro HC},
  journal={Journal of econometrics},
  volume={225},
  number={2},
  pages={200--230},
  year={2021},
  publisher={Elsevier}
}

@article{li2020rerandomization,
	title={Rerandomization and regression adjustment},
	author={Li, Xinran and Ding, Peng},
	journal={Journal of the Royal Statistical Society Series B: Statistical Methodology},
	volume={82},
	number={1},
	pages={241--268},
	year={2020},
	publisher={Oxford University Press}
}

@book{imbens2015causal,
  title={Causal inference in statistics, social, and biomedical sciences},
  author={Imbens, Guido W and Rubin, Donald B},
  year={2015},
  publisher={Cambridge university press}
}

@article{li2024evaluating,
  title={Evaluating dynamic conditional quantile treatment effects with applications in ridesharing},
  author={Li, Ting and Shi, Chengchun and Lu, Zhaohua and Li, Yi and Zhu, Hongtu},
  journal={Journal of the American Statistical Association},
  volume={119},
  number={547},
  pages={1736--1750},
  year={2024},
  publisher={Taylor \& Francis}
}

@article{hu2022switchback,
  title={Switchback experiments under geometric mixing},
  author={Hu, Yuchen and Wager, Stefan},
  journal={arXiv preprint arXiv:2209.00197},
  year={2022}
}

@article{bojinov2023design,
  title={Design and analysis of switchback experiments},
  author={Bojinov, Iavor and Simchi-Levi, David and Zhao, Jinglong},
  journal={Management Science},
  volume={69},
  number={7},
  pages={3759--3777},
  year={2023},
  publisher={INFORMS}
}

@article{wang2018quantile,
	title={Quantile-optimal treatment regimes},
	author={Wang, Lan and Zhou, Yu and Song, Rui and Sherwood, Ben},
	journal={Journal of the American Statistical Association},
	volume={113},
	number={523},
	pages={1243--1254},
	year={2018},
	publisher={Taylor \& Francis}
}

@article{rubin2005causal,
	title={Causal inference using potential outcomes: Design, modeling, decisions},
	author={Rubin, Donald B},
	journal={Journal of the American statistical Association},
	volume={100},
	number={469},
	pages={322--331},
	year={2005},
	publisher={Taylor \& Francis}
}

@inproceedings{yang2018mean,
	title={Mean field multi-agent reinforcement learning},
	author={Yang, Yaodong and Luo, Rui and Li, Minne and Zhou, Ming and Zhang, Weinan and Wang, Jun},
	booktitle={International conference on machine learning},
	pages={5571--5580},
	year={2018},
	organization={PMLR}
}

@article{li2022random,
	title={Random graph asymptotics for treatment effect estimation under network interference},
	author={Li, Shuangning and Wager, Stefan},
	journal={The Annals of Statistics},
	volume={50},
	number={4},
	pages={2334--2358},
	year={2022},
	publisher={Institute of Mathematical Statistics}
}

@article{shi2023multiagent,
	title={A multiagent reinforcement learning framework for off-policy evaluation in two-sided markets},
	author={Shi, Chengchun and Wan, Runzhe and Song, Ge and Luo, Shikai and Zhu, Hongtu and Song, Rui},
	journal={The Annals of Applied Statistics},
	volume={17},
	number={4},
	pages={2701--2722},
	year={2023},
	publisher={Institute of Mathematical Statistics}
}

@article{su2021model,
	title={Model-assisted analyses of cluster-randomized experiments},
	author={Su, Fangzhou and Ding, Peng},
	journal={Journal of the Royal Statistical Society Series B: Statistical Methodology},
	volume={83},
	number={5},
	pages={994--1015},
	year={2021},
	publisher={Oxford University Press}
}

@article{zhu2014spatially,
	title={Spatially varying coefficient model for neuroimaging data with jump discontinuities},
	author={Zhu, Hongtu and Fan, Jianqing and Kong, Linglong},
	journal={Journal of the American Statistical Association},
	volume={109},
	number={507},
	pages={1084--1098},
	year={2014},
	publisher={Taylor \& Francis}
}

@article{chernozhukov2013gaussian,
	title={Gaussian approximations and multiplier bootstrap for maxima of sums of high-dimensional random vectors},
	author={Chernozhukov, Victor and Chetverikov, Denis and Kato, Kengo},
	journal={Annals of Statistics},
	volume={41},
	number={6},
	pages={2786--2819},
	year={2013}
}

@article{luo2024policy,
  title={Policy evaluation for temporal and/or spatial dependent experiments},
  author={Luo, Shikai and Yang, Ying and Shi, Chengchun and Yao, Fang and Ye, Jieping and Zhu, Hongtu},
  journal={Journal of the Royal Statistical Society Series B: Statistical Methodology},
  volume={86},
  number={3},
  pages={623--649},
  year={2024},
  publisher={Oxford University Press UK}
}

@misc{shumway2006time,
  title={Time Series Analysis and Its Applications: With R Examples},
  author={Shumway, Rober H},
  year={2006},
  publisher={Springer}
}

@article{hastie1993varying,
  title={Varying-coefficient models},
  author={Hastie, Trevor and Tibshirani, Robert},
  journal={Journal of the Royal Statistical Society Series B: Statistical Methodology},
  volume={55},
  number={4},
  pages={757--779},
  year={1993},
  publisher={Oxford University Press}
}

@article{hoover1998nonparametric,
  title={Nonparametric smoothing estimates of time-varying coefficient models with longitudinal data},
  author={Hoover, Donald R and Rice, John A and Wu, Colin O and Yang, Li-Ping},
  journal={Biometrika},
  volume={85},
  number={4},
  pages={809--822},
  year={1998},
  publisher={Oxford University Press}
}

@article{fan1999statistical,
  title={Statistical estimation in varying coefficient models},
  author={Fan, Jianqing and Zhang, Wenyang},
  journal={The annals of Statistics},
  volume={27},
  number={5},
  pages={1491--1518},
  year={1999},
  publisher={Institute of Mathematical Statistics}
}

@article{chiang2001smoothing,
  title={Smoothing spline estimation for varying coefficient models with repeatedly measured dependent variables},
  author={Chiang, Chin-Tsang and Rice, John A and Wu, Colin O},
  journal={Journal of the American Statistical Association},
  volume={96},
  number={454},
  pages={605--619},
  year={2001},
  publisher={Taylor \& Francis}
}

@article{huang2004polynomial,
  title={Polynomial spline estimation and inference for varying coefficient models with longitudinal data},
  author={Huang, Jianhua Z and Wu, Colin O and Zhou, Lan},
  journal={Statistica Sinica},
  pages={763--788},
  year={2004},
  publisher={JSTOR}
}

@article{park2015varying,
  title={Varying coefficient regression models: a review and new developments},
  author={Park, Byeong U and Mammen, Enno and Lee, Young K and Lee, Eun Ryung},
  journal={International Statistical Review},
  volume={83},
  number={1},
  pages={36--64},
  year={2015},
  publisher={Wiley Online Library}
}

@article{sun2023semiparametric,
  title={Semiparametric additive time-varying coefficients model for longitudinal data with censored time origin},
  author={Sun, Yanqing and Shou, Qiong and Gilbert, Peter B and Heng, Fei and Qian, Xiyuan},
  journal={Biometrics},
  volume={79},
  number={2},
  pages={695--710},
  year={2023},
  publisher={Wiley Online Library}
}

@article{zhang2025bayesian,
  title={A Bayesian spatial--temporal varying coefficients model for estimating excess deaths associated with respiratory infections},
  author={Zhang, Yuzi and Chang, Howard H and Iuliano, Angela D and Reed, Carrie},
  journal={Journal of the Royal Statistical Society Series A: Statistics in Society},
  volume={188},
  number={3},
  pages={843--858},
  year={2025},
  publisher={Oxford University Press UK}
}

@inproceedings{sadovnichiy2022mathematical,
  title={Mathematical modeling of overcoming the COVID-19 pandemic and restoring economic growth},
  author={Sadovnichiy, VA and Akaev, Askar Akayevich and Zvyagintsev, AI and Sarygulov, Askar Islamovich},
  booktitle={Doklady Mathematics},
  volume={106},
  number={1},
  pages={230--235},
  year={2022},
  organization={Springer}
}

@article{rubin1980randomization,
  title={Randomization analysis of experimental data: The Fisher randomization test comment},
  author={Rubin, Donald B},
  journal={Journal of the American statistical association},
  volume={75},
  number={371},
  pages={591--593},
  year={1980},
  publisher={JSTOR}
}

@article{halloran2016dependent,
  title={Dependent happenings: a recent methodological review},
  author={Halloran, M Elizabeth and Hudgens, Michael G},
  journal={Current epidemiology reports},
  volume={3},
  number={4},
  pages={297--305},
  year={2016},
  publisher={Springer}
}

@article{savje2021average,
  title={Average treatment effects in the presence of unknown interference},
  author={S{\"a}vje, Fredrik and Aronow, Peter and Hudgens, Michael},
  journal={Annals of statistics},
  volume={49},
  number={2},
  pages={673},
  year={2021}
}

@inproceedings{farias2023correcting,
  title={Correcting for interference in experiments: A case study at douyin},
  author={Farias, Vivek and Li, Hao and Peng, Tianyi and Ren, Xinyuyang and Zhang, Huawei and Zheng, Andrew},
  booktitle={Proceedings of the 17th ACM Conference on Recommender Systems},
  pages={455--466},
  year={2023}
}

@inproceedings{liu2023unite,
  title={UniTE: A Unified Treatment Effect Estimation Method for One-sided and Two-sided Marketing},
  author={Liu, Runshi and Hou, Zhipeng},
  booktitle={Proceedings of the 32nd ACM International Conference on Information and Knowledge Management},
  pages={1472--1481},
  year={2023}
}

@article{johari2022experimental,
  title={Experimental design in two-sided platforms: An analysis of bias},
  author={Johari, Ramesh and Li, Hannah and Liskovich, Inessa and Weintraub, Gabriel Y},
  journal={Management Science},
  volume={68},
  number={10},
  pages={7069--7089},
  year={2022},
  publisher={INFORMS}
}

@article{li2025testing,
  title={Testing stationarity and change point detection in reinforcement learning},
  author={Li, Mengbing and Shi, Chengchun and Wu, Zhenke and Fryzlewicz, Piotr},
  journal={The Annals of Statistics},
  volume={53},
  number={3},
  pages={1230--1256},
  year={2025},
  publisher={Institute of Mathematical Statistics}
}

@inproceedings{wang2023robust,
  title={A robust test for the stationarity assumption in sequential decision making},
  author={Wang, Jitao and Shi, Chengchun and Wu, Zhenke},
  booktitle={International Conference on Machine Learning},
  pages={36355--36379},
  year={2023},
  organization={PMLR}
}

@article{li2023optimal,
  title={Optimal treatment allocation for efficient policy evaluation in sequential decision making},
  author={Li, Ting and Shi, Chengchun and Wang, Jianing and Zhou, Fan and others},
  journal={Advances in Neural Information Processing Systems},
  volume={36},
  pages={48890--48905},
  year={2023}
}

@article{basse2018limitations,
  title={Limitations of design-based causal inference and a/b testing under arbitrary and network interference},
  author={Basse, Guillaume W and Airoldi, Edoardo M},
  journal={Sociological Methodology},
  volume={48},
  number={1},
  pages={136--151},
  year={2018},
  publisher={SAGE Publications Sage CA: Los Angeles, CA}
}

@article{liu2024cluster,
  title={Cluster-adaptive network a/b testing: From randomization to estimation},
  author={Liu, Yang and Zhou, Yifan and Li, Ping and Hu, Feifang},
  journal={Journal of Machine Learning Research},
  volume={25},
  number={170},
  pages={1--48},
  year={2024}
}

@article{masoero2026multiple,
  title={Multiple randomization designs: Estimation and inference with interference},
  author={Masoero, Lorenzo and Vijaykumar, Suhas and Richardson, Thomas S and McQueen, James and Rosen, Ido and Burdick, Brian and Bajari, Pat and Imbens, Guido},
  journal={Journal of the Royal Statistical Society Series B: Statistical Methodology},
  pages={qkaf073},
  year={2026},
  publisher={Oxford University Press UK}
}

@article{chen2024real,
  title={Real-time spatial-intertemporal dynamic pricing for balancing supply and demand in a ride-hailing network: near-optimal policies and the value of dynamic pricing},
  author={Chen, Q and Lei, Y and Jasin, Stefanus},
  journal={Operations Research},
  volume={72},
  number={5},
  pages={2097--2118},
  year={2024},
  publisher={INFORMS (Institute for Operations Research and Management Sciences)}
}

@article{ashenfelter1978estimating,
  title={Estimating the effect of training programs on earnings},
  author={Ashenfelter, Orley},
  journal={The Review of Economics and Statistics},
  pages={47--57},
  year={1978},
  publisher={JSTOR}
}

\newpage
\appendix

\section*{}
\vskip -0.2in
\hrule height 1.5pt
\vskip 0.4cm
\begin{center}
    \textbf{\LARGE Appendix}
\end{center}
\vskip 0.15cm
\hrule height 1.5pt
\vskip 1cm

\section{The proposed algorithm}
\label{sec:algorithm}

\begin{algorithm}[htbp]
\caption{Estimation and Inference for GATE}
\label{alg:estimation and testing}
\begin{algorithmic}[1]

\Require Observed aggregate data
$\{Y_{d,C}(t),\, S_{d,C}(t),\, X_{d,C}(t),\, \widetilde{X}_{d,C}(t),\, f(t)\}$
for $d = 1,\dots,n$, $t = 1,\dots,m$, $C \in \{C_0, C_1\}$;
number of bootstrap samples $B$; significance level $\alpha$.

\State
Estimate the time-varying coefficients in the outcome model~\eqref{eq:model_Y} and the state
model~\eqref{eq:model_state_within_group_all} via two-step kernel smoothing, yielding
$\widetilde{\theta}_C(t)$ and $\widetilde{\Theta}_C$ as in~\eqref{eq:coe_est}.
Substitute these into~\eqref{eq:GATE_coe_expression} to obtain $\widehat{\text{GATE}}$.

\State
For each $d = 1, \dots, n$ and $C \in \{C_0, C_1\}$, compute the fitted values
$\widetilde{S}_{d,C}(t+1)$ and $\widetilde{Y}_{d,C}(t)$, along with the residuals:
\begin{align*}
    \widehat{\epsilon}_{d,C}(t)
        &= Y_{d,C}(t) - Z_{d,C}(t)^\top \widetilde{\theta}_C(t),
        \quad t = 1,\dots,m, \\
    \widehat{E}_{d,C}(t+1)
        &= S_{d,C}(t+1) - \widetilde{\Theta}_C\,\widetilde{Z}_{d,C}(t),
        \quad t = 1,\dots,m-1.
\end{align*}

\For{$b = 1$ \textbf{to} $B$}

    \State
    Draw i.i.d.\ standard normal multipliers
    $\bigl\{\xi_{d,Y,C_0}^b,\, \xi_{d,Y,C_1}^b,\,
             \xi_{d,S,C_0}^b,\, \xi_{d,S,C_1}^b\bigr\}_{d=1}^n$.

    \State
    Construct bootstrap pseudo-responses by perturbing the fitted values with the
    multiplier-weighted residuals:
    \begin{align*}
        \widehat{S}_{d,C}^{\,b}(t+1)
            &= \widetilde{S}_{d,C}(t+1)
               + \xi_{d,S,C}^b\,\widehat{E}_{d,C}(t+1), \\
        \widehat{Y}_{d,C}^{\,b}(t)
            &= \widetilde{Y}_{d,C}(t)
               + \xi_{d,Y,C}^b\,\widehat{\epsilon}_{d,C}(t).
    \end{align*}

    \State
    Re-estimate the model using the pseudo-data
    $\{\widehat{S}_{d,C}^{\,b}(t),\, \widehat{Y}_{d,C}^{\,b}(t)\}$
    to obtain $\widetilde{\theta}_C^{\,b}(t)$ and $\widetilde{\Theta}_C^{\,b}$.

    \State
    Compute the centered bootstrap statistic:
    \[
        T^b = \widehat{\text{GATE}}^{\,b} - \widehat{\text{GATE}}.
    \]

\EndFor

\State
Reject $H_0$ at significance level $\alpha$ if
\[
    T \;=\; \widehat{\text{GATE}}
    \;>\;
    \text{the } (1-\alpha)\text{-quantile of }
    \bigl\{T^b\bigr\}_{b=1}^B.
\]

\Ensure
Estimate $\widehat{\text{GATE}}$ and the corresponding decision rule for testing $H_0$.

\end{algorithmic}
\end{algorithm}

\section{Additional Experiment Results}
\label{appendix:additional_experiment}
In this section, we describe the data-generating process for the real-data--based simulation and report additional experimental results.

\textbf{Example \ref{ex:simulation_single_city} (continued).} 
The A/A experiment spans two weeks, 
with each day divided into $m=24$ time intervals.
For the GMV outcome model, in addition to the supply and demand state variables, the predictive covariates include the other subsidy amount and the supply--demand gap $G_d(t)$. The state-transition model uses temperature, precipitation, and a holiday indicator (equal to 1 if the time interval falls on a public holiday, and 0 otherwise) as predictors. We first fit models \eqref{eq:model_Y}--\eqref{eq:model_state_within_group_all} using the full A/A dataset to obtain coefficient estimates $\{ \widetilde{\alpha}_0,\widetilde{\alpha}_1, \widetilde{\alpha}_2 \}$ and $\{ \widetilde{\gamma}_0(t), \widetilde{\Phi}_0(t), \widetilde{\Phi}_1(t) \}$, along with estimated error processes $\widetilde{e}_d(t)$ and $\widetilde{E}_d(t)$.

To generate simulated data, we adopt a bootstrap procedure. In each run, we sample $n \in \{14, 28\}$ initial state observations and $n$ estimated error trajectories with replacement. For each group $C \in \{C_0, C_1\}$, we sample the initial demand $\widetilde D_{d,C}(1)$ within the group and the initial supply $\widetilde S_d(1)$ from platform-level supply. For each city and time interval $t$, predictive covariates are generated as follows. In the outcome model, the other subsidy amount and the supply--demand gap $G_d(t)$ are independently drawn from uniform distributions over their empirical ranges at time $t$ during the A/A period; the subsidy is sampled separately for treatment and control groups, while $G_d(t)$ is shared. In the state-transition model, temperature and precipitation are sampled analogously and shared across groups, while the holiday indicator follows the calendar. We then generate $n$ days of data under the VCDP model:
\begin{eqnarray*}
\widetilde Y_{d, C} (t) &=& 
\alpha_{0,C} ( t) + \alpha_{1,C}^\top (t) X_{d, C}(t)
+ \alpha_{2, C}^\top (t) \widetilde S_{d,C}(t) + \widetilde\epsilon_{d,C}(t), \\
\widetilde S_{d, C}(t+1) &=&  \gamma_{0, C}(t) + \Phi_{0, C}(t) \widetilde{X}_{d, C}(t)
+ \Phi_{1, C}(t) \widetilde S_{d, C} (t) + \widetilde E_{d, C}(t+1).
\end{eqnarray*}

We set $\{ \alpha_{0,C}, \alpha_{1,C}, \alpha_{2,C} \} = \{ \widetilde{\alpha}_0, \widetilde{\alpha}_1, \widetilde{\alpha}_2 \}$ for both groups $C \in \{C_0, C_1\}$. The parameters $\gamma_{0,C}(t)$ and $\Phi_{1,C}(t)$ are shared across groups and set equal to $\{ \widetilde{\gamma}_0(t), \widetilde{\Phi}_1(t) \}$. To control the magnitude of the GATE, we let $\Phi_{0,C_0}(t) = \widetilde{\Phi}_0(t)$ and $\Phi_{0,C_1}(t) = \eta \widetilde{\Phi}_0(t)$ for a constant $\eta \ge 0$. When $\eta = 0$, there is no treatment effect; when $\eta > 0$, the new policy is beneficial. 
For each city and each value of $\eta$, we generate 1000 independent datasets and apply the proposed test, the two-sample $t$-test, the DiD estimator, and the DE estimator at the 5\% significance level.

Figure \ref{fig:city_pvalues} shows that the empirical distributions of bootstrap $p$-values under the null across the three cities closely approximate the uniform distribution, confirming the validity of the proposed testing procedure.

\begin{figure}[htbp]
\vskip -0.15in
    \centering
    \includegraphics[width=0.9\columnwidth]{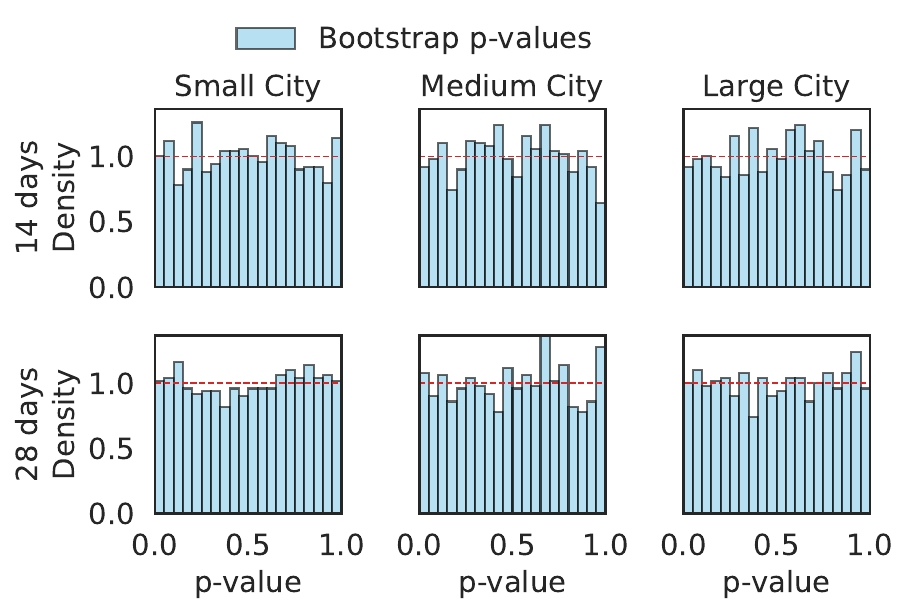}
    \caption{Empirical distributions of bootstrap $p$-values under the null for three cities.}
    \label{fig:city_pvalues}
\vskip -0.15in
\end{figure}

{\bf Example \ref{ex:simulation_all_city} (Continued).}
This example mimics a setting where the experiment is conducted simultaneously across $R$ cities. In each city, treatment and control groups receive different subsidy amounts under a common budget constraint. The objective is to evaluate whether the new subsidy allocation policy improves the overall GMV across the $R$ cities. Participants within each city are randomly selected to participate in the experiment.
Let $\{N_r,\, r=1,\dots,R\}$ denote the number of participants in the $r$-th city, and let $N=\sum_{r=1}^R N_r$. The GATE is defined as
\begin{equation}
\label{eq:GATE_def_individual_city}
\begin{aligned}
\text{GATE}
= N^{-1} \sum_{r=1}^R 
\Big[
  \mathbb{E} \Big( \sum_{t=1}^m \sum_{i=1}^{N_r} Y^i(t, \bm{1}_{N,t}) \Big)
  - \mathbb{E} \Big( \sum_{t=1}^m \sum_{i=1}^{N_r} Y^i(t, \bm{0}_{N,t}) \Big)
\Big].
\end{aligned}
\end{equation}

The data-generating process is the same as in Example~\ref{ex:simulation_single_city}: data are generated independently within each city following the same procedure. However, the GATE is evaluated jointly across all cities. Specifically, we first estimate the GATE within each city, and then aggregate these estimates to obtain the overall GATE. The bootstrap-based test is constructed accordingly.

As shown in Figure~\ref{fig:joint_result}, the distribution of the bootstrap statistic closely approximates the empirical sampling distribution of the estimated GATE. Under the null hypothesis, the empirical distribution of bootstrap $p$-values remains approximately uniform on $[0,1]$, indicating good calibration. 
Moreover, the proposed test demonstrates a clear power advantage over all three alternative methods, particularly for small effect sizes.

\begin{figure}[htbp]
    \centering
    \includegraphics[width=0.9\columnwidth]{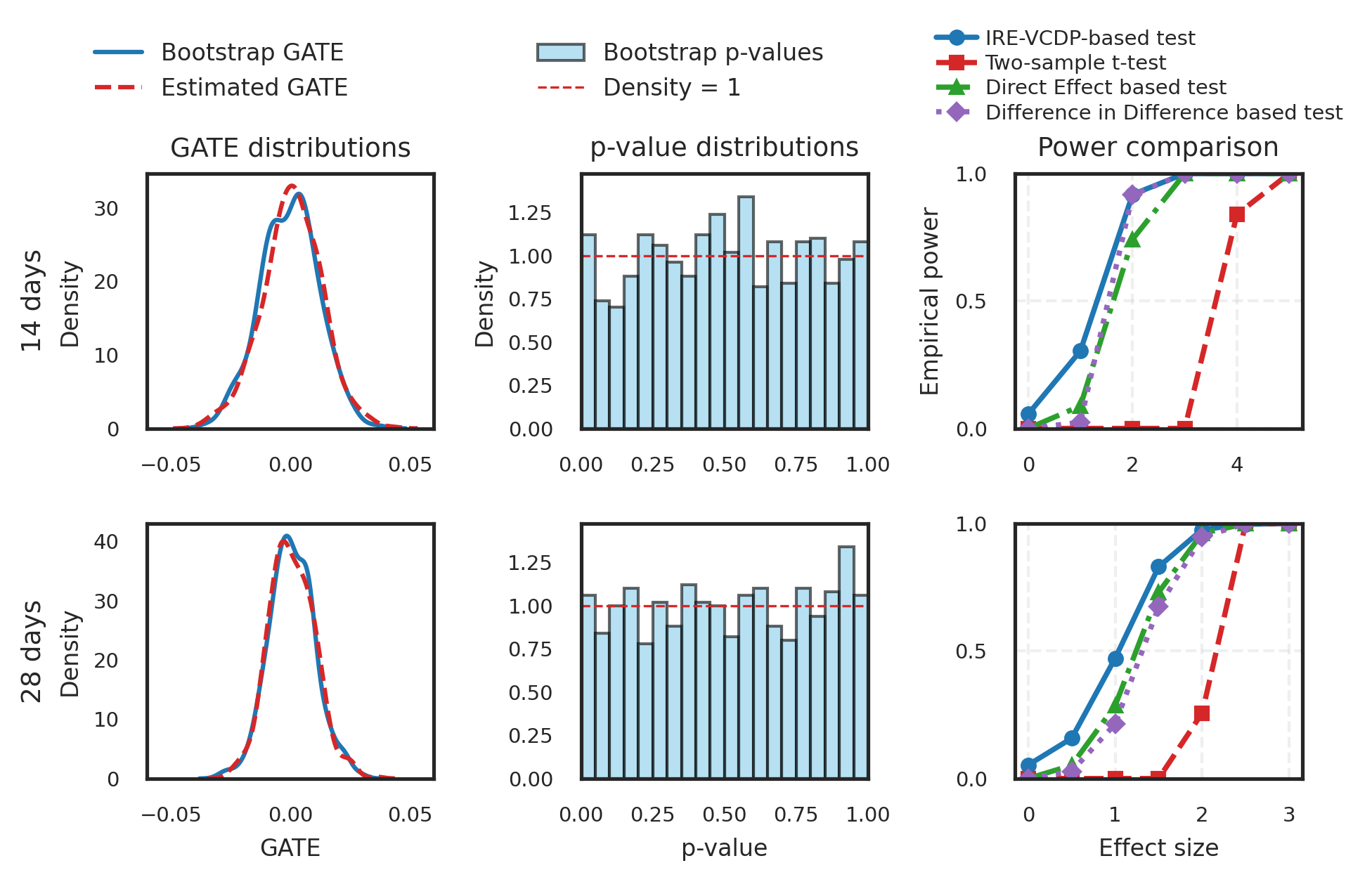}
    \caption{Joint simulation results pooling all three cities. Columns report GATE sampling distributions, bootstrap $p$-value distributions under the null, and power curves of the proposed test versus the $t$-test, the DiD estimator, and the DE estimator as a function of the true GATE; rows correspond to the two observation horizons ($n=14$ and $n=28$ days).} 
    \label{fig:joint_result}
\end{figure}

{\bf Example \ref{ex:real_data} (Continued).} Figure~\ref{fig:cities_three_plots} illustrates the temporal evolution of GMV and key state variables capturing supply--demand dynamics (number of requests and drivers’ total online time) across cities over the experimental period. Figures~\ref{fig:city_53_res_plots_large}--\ref{fig:city_310_res_plots_small} present the fitted values of GMV, number of requests, and drivers’ total online time against their observed counterparts, together with the corresponding residuals over time for both the small and large cities. These results demonstrate that the proposed VCDP model provides a good fit to both the outcome and state variables.

\begin{figure}[htbp]
    \centering
    \includegraphics[width=0.9\columnwidth]{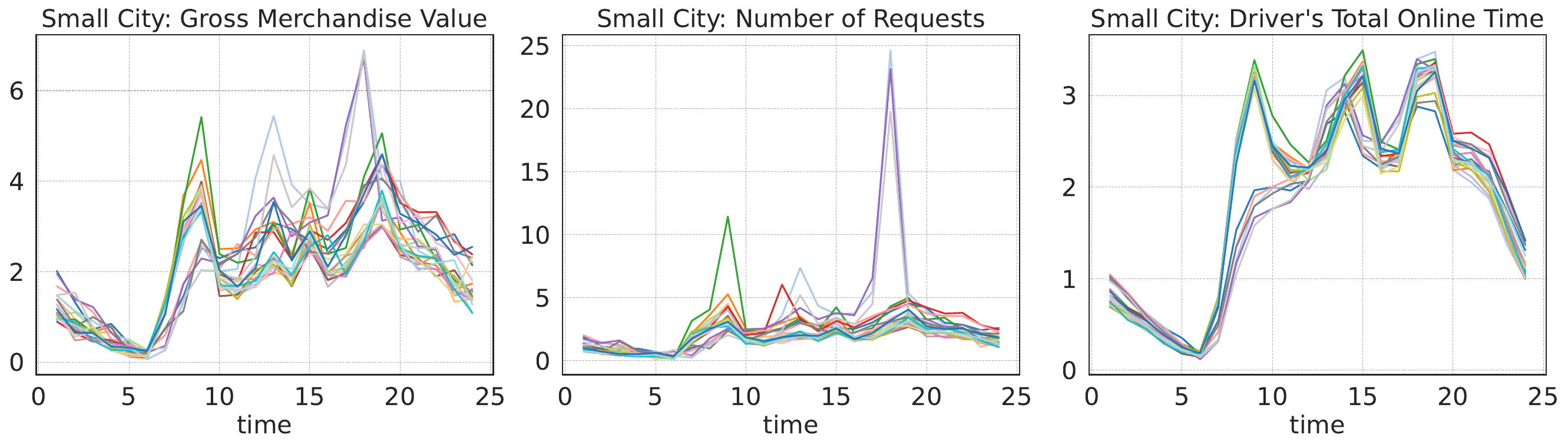} 
        \includegraphics[width=0.9\columnwidth]{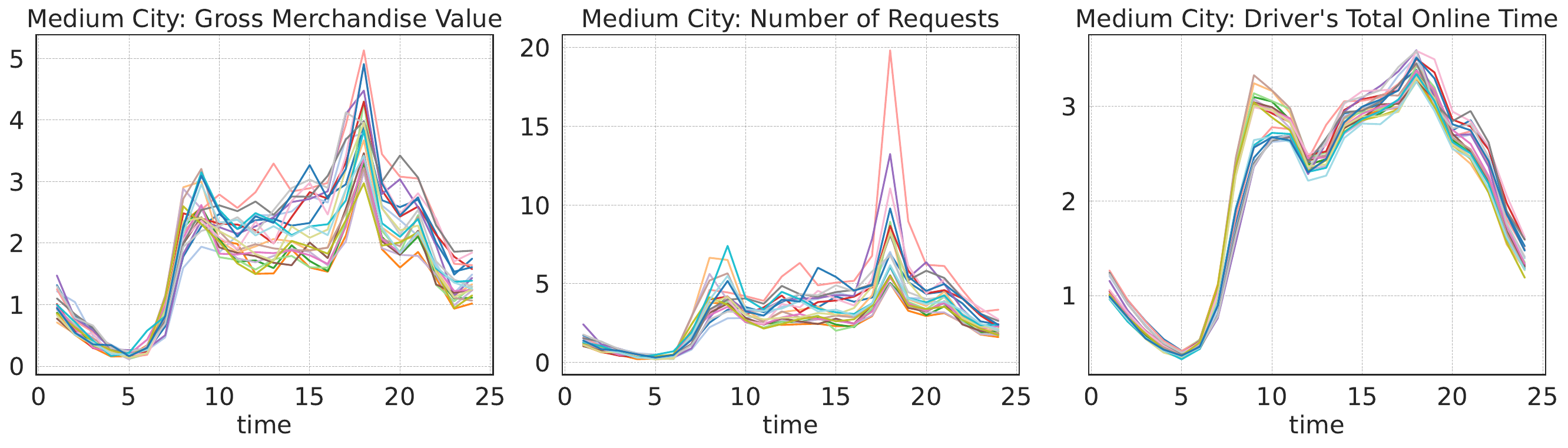}
    \includegraphics[width=0.9\columnwidth]{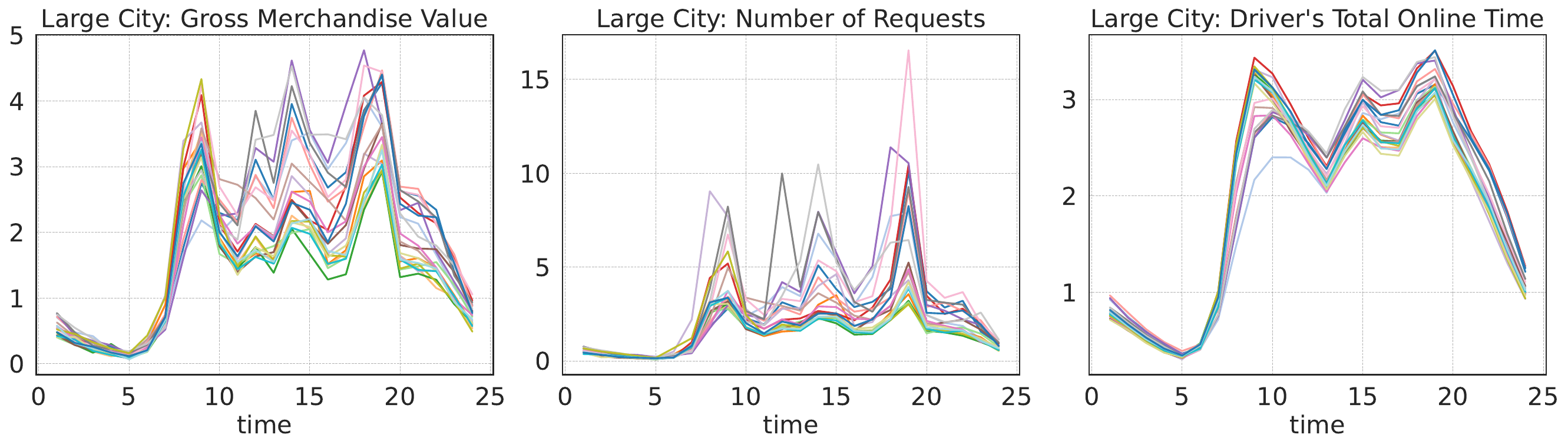}
    \caption{Scaled business metrics from the small, medium and large cities across 21 days.}
    \label{fig:cities_three_plots}
\end{figure}

\begin{figure}[htbp]
    \centering
    \includegraphics[width=0.9\columnwidth]{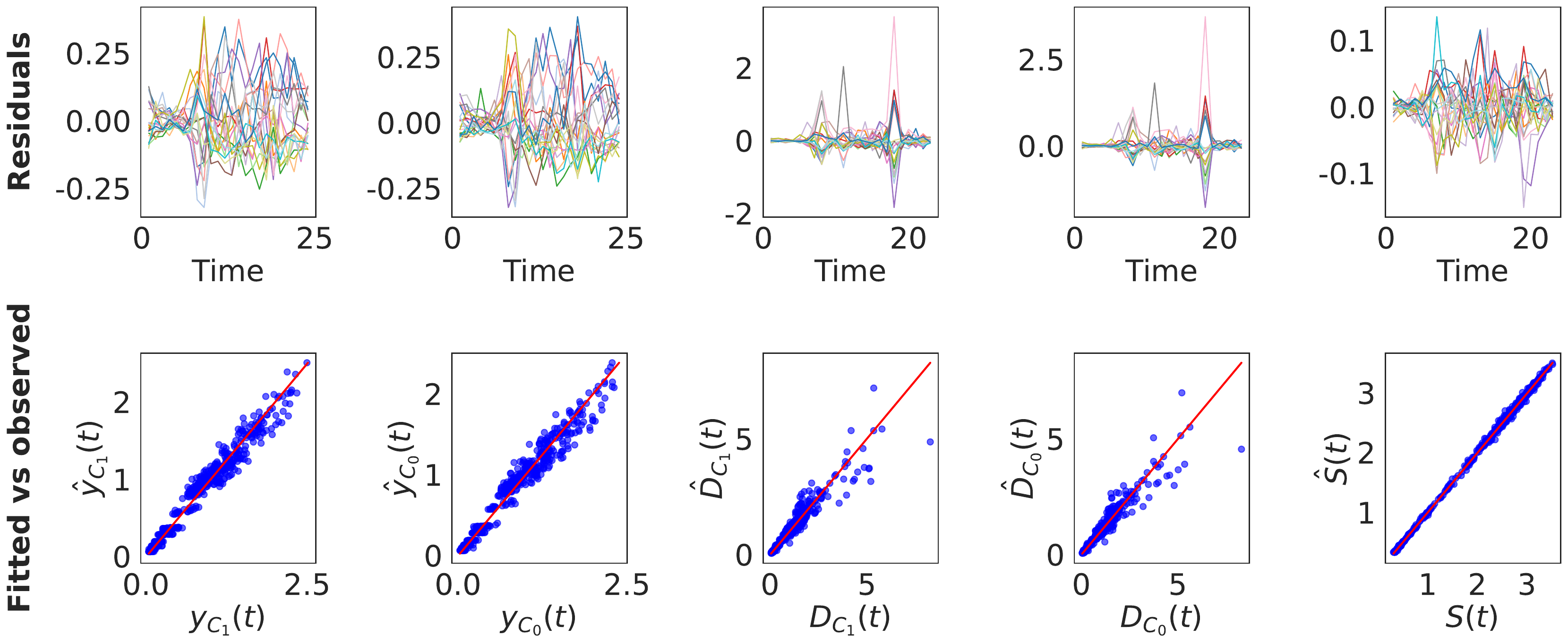}
    \caption{Large city: residuals (top row) and fitted vs. observed values (bottom row). Columns (left to right): GMV (treatment), GMV (control), number of requests (treatment), number of requests (control), and drivers’ total online time.}
     \label{fig:city_53_res_plots_large}
\end{figure}

\begin{figure}[htbp]
    \centering
    \includegraphics[width=0.9\columnwidth]{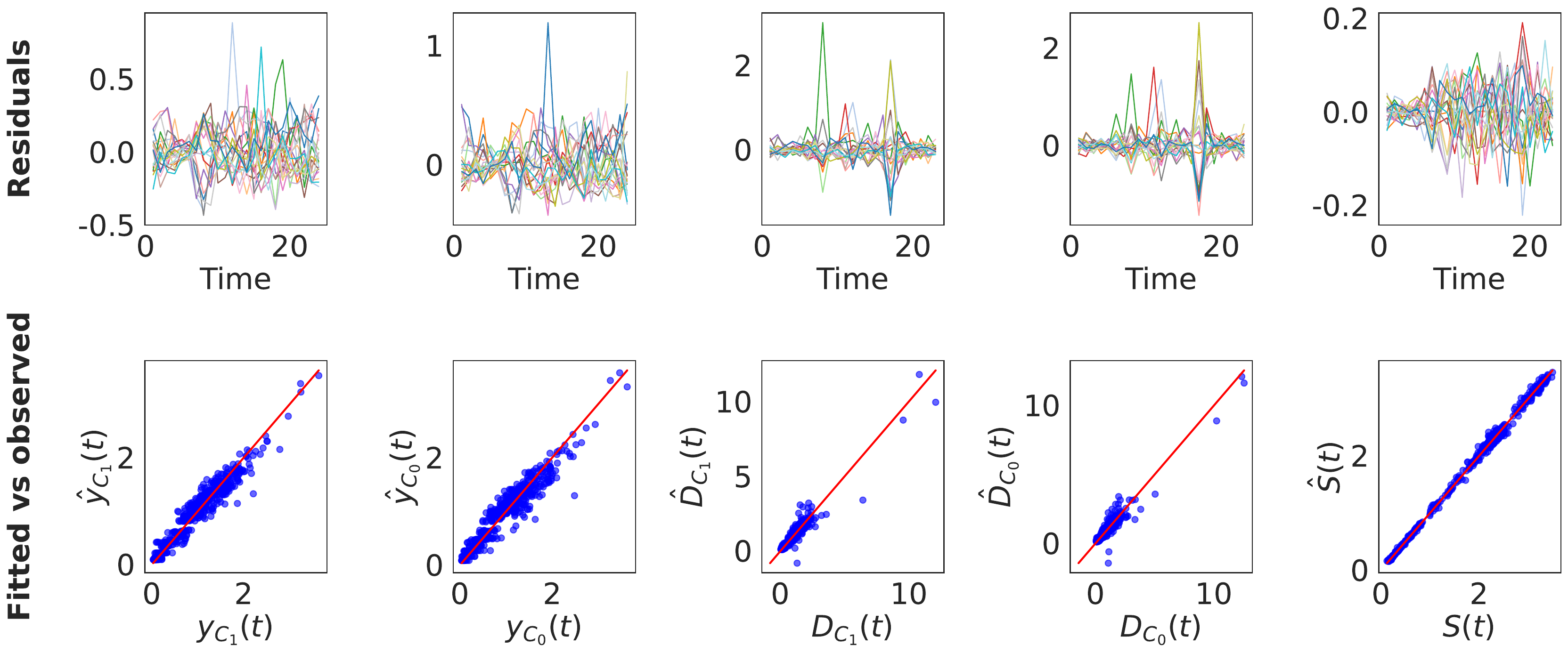}
    \caption{Small city: residuals (top row) and fitted vs. observed values (bottom row). Columns (left to right): GMV (treatment), GMV (control), number of requests (treatment), number of requests (control), and drivers’ total online time.}
     \label{fig:city_310_res_plots_small}
\end{figure}

Table~\ref{tab:pvalues} reports the $p$-values from all four methods for A/A and A/B analyses across cities. Detailed discussion is provided in Section~\ref{sec:experiments}.

\begin{table}[htbp]
 \caption{P-value results for different methods and cities under A/A and A/B periods.}
  \label{tab:pvalues}
  \centering
  \begin{tabular}{lcccccccc}
      \toprule
      Method & \multicolumn{2}{c}{Small City} & \multicolumn{2}{c}{Medium City} & \multicolumn{2}{c}{Large City} &
\multicolumn{2}{c}{All Cities} \\

      \cmidrule(lr){2-3} \cmidrule(lr){4-5} \cmidrule(lr){6-7} \cmidrule(lr){8-9}

       & AA & AB & AA & AB & AA & AB & AA & AB \\
      \midrule

      IRE-VCDP & 0.382 & 0.078 & 0.544 & 0.244 & 0.488 & 0.208 & 0.446 & 0.12 \\
      Two-sample t-test & 0.485 & 0.476 & 0.397 & 0.525 & 0.405 & 0.415 & 0.391 & 0.553 \\
      DE & 0.5 & 0.472 & 0.248 & 0.54 & 0.252 & 0.306 & 0.158 & 0.534 \\
      DiD & 0.515 & 0.362 & 0.099 & 0.631 & 0.532 & 0.253 & 0.366 & 0.429 \\
      \bottomrule
  \end{tabular}

\end{table}

\section{Notation and Definitions}

\subsection{Formulation and Discussion of Straightforward Treatment Effect}\label{app:Formulation of Straightforward Treatment Effect}
\begin{align*}
\tau
=&
\sum_{t=1}^m 
(  \alpha_{0,C_1} (t) -  \alpha_{0,C_0}(t) )
+ \sum_{t=1}^m \Big[ \mathbb{E}( X_{C_1} (t)  )  \alpha_{1,C_1} (t)  -
\mathbb{E}( X_{C_0} (t)  )   \alpha_{1,C_0} (t)  \Big]  \\
&  +  \sum_{t=1}^m \Big[\mathbb{E}( S_{C_1}(t)  ) 
\alpha_{2,C_1} (t) - \mathbb{E}( S_{C_0}(t)  )\alpha_{2,C_0} (t)
\Big].
\end{align*}

If the means of the predictive covariates are equal across groups, such that $ \mathbb{E}( N^{-1} \sum_{i=1}^N X^i(t)  )= \mathbb{E}( X_{C_1} (t)  )  = \mathbb{E}( X_{C_0} (t)  ) $, and $\mathbb{E}( S(t, \bm{1}_{N,t})  )= \mathbb{E}( S(t, \bm{0}_{N,t})  ) $, then $\tau$ coincides with the GATE. In practice, ensuring that the mean of the predictive covariates is balanced between the two groups is relatively straightforward using covariate balancing techniques such as matching or weighting. However, achieving equality in the expected potential outcomes of the state variables is more challenging, as the treatment history often directly influences their evolution over time.

\subsection{Formulation of GATE Estimator}\label{app:Formulation of GATE Estimator}

\begin{align*}
	&\widehat{\text{GATE}} =
\sum_{t=1}^m 
(  \widetilde{\alpha}_{0,C_1} (t) -  \widetilde{\alpha}_{0,C_0}(t) )
+
\sum_{t=1}^m \Big[ n^{-1} \sum_{d=1}^n X_d(t) ( \widetilde{\alpha}_{1,C_1} (t)  -  \widetilde{\alpha}_{1,C_0} (t) )  \Big] \\ 
&+\sum_{t=1}^m \Big[
\widetilde{\alpha}_{2,C_1} (t)^\top   \Big( \sum_{k=1}^{t-1} \{ \prod_{l=k+1}^{t-1} \widetilde{\Phi}_{1,C_1}(l) \Big[\widetilde{\gamma}_{0, C_1}(k)+\widetilde{\gamma}_{1, C_1}(k)+n^{-1}\widetilde{\Phi}_{0, C_1} (t)\sum_{d=1}^nH_{d,C_1}(t) \Big]  \}
+ \prod_{l=1}^{t-1} \widetilde{\Phi}_{1,C_1}(l) \mathbb{E}(S(1)) 
\Big)
\\
&-
\widetilde{\alpha}_{2,C_0} (t)^\top \Big(
\sum_{k=1}^{t-1} \{ \prod_{l=k+1}^{t-1} \widetilde{\Phi}_{1, C_0}(l) \Big[\widetilde\gamma_{0, C_0}(k)+n^{-1}\widetilde{\Phi}_{0, C_0} (t)\sum_{d=1}^nH_{d,C_0}(t) \Big] \}
+ \prod_{l=1}^{t-1} \widetilde{\Phi}_{1, C_0}(l) \mathbb{E}(S(1))\Big)
\Big], 
\end{align*}
where $X_d(t):=N^{-1}(N_0 X_{d,C_0}(t)+ N_1 X_{d,C_1}(t))$.

\section{Assumptions for Theorem \ref{thm:bootstrap_consistency_QTE}}
In this section,  we provide the regularity conditions for the bootstrap consistency theory.

\begin{assumption}
	\label{assump:kernel}
	The kernel function $K(\cdot)$ is a symmetric probability density function defined over the interval $ [-1,1] $. It is Lipschitz continuous and satisfies the condition $\int_{-1}^1 | t K^\prime(t )| dt < \infty$, signifying its finiteness under the integral of its absolute derivative. 
\end{assumption}

\begin{assumption}
	\label{assump:Z}
    The covariates used in outcome model and state model are independent and identically distributed across $d$ and are sub-Gaussian (uniformly over $t$). Furthermore, for any integer $t$ within the range $1\leq t\leq m$, the smallest eigenvalue of each population second-moment matrix of the covariates is bounded away from zero.
\end{assumption}

\begin{assumption}
	\label{assump:coe}
	All components of $\theta_{C_0} (ms)$, $\theta_{C_1} (ms)$, $\Theta_{C_0}(ms)$ and $\Theta_{C_1}(ms)$ possess second-order derivatives with respect to $s$. 
\end{assumption}	

\begin{assumption}
	\label{asmp:st1}
	There exist constants $0<c_1<1$, $M_{\Gamma},M_{\beta}>0$, and $M_{min}>0$ such that   $\Vert\Phi_{1,C}(t)\Vert_\infty\leq c_1$, $\Vert\alpha_{2,C}(t)\Vert_\infty\leq M_{\alpha}$, and $\Vert\Phi_{0,C}(t)\Vert_\infty\leq M_\Phi$, for $\forall C\in\{C_0.C_1\}$.
\end{assumption}

Assumption \ref{assump:kernel} is mild as the kernel $K(\cdot)$ is user-specified. Assumption \ref{assump:Z} has been commonly used
in the literature on varying coefficient models (see e.g., \citet{zhu2014spatially}). Assumption \ref{assump:coe} serves as a routine smoothness condition that guarantees the local approximation accuracy of the kernel estimator. Assumption \ref{asmp:st1} ensures that
the time series is stationary, since $\Phi(t)$ is the autoregressive coefficient. It is commonly imposed in the literature on time series analysis \citep{shumway2006time}.

\section{Proofs of Propositions and Theorems}

\subsection{Proof of Proposition \ref{prop: identification}}\label{app:proof of identification}

\begin{proof}
The proof proceeds in two steps. First, we show that the function $R(\cdot)$ can be identified from the observed data using the conditional expectation of the observed outcomes. Second, we apply the law of iterated expectations to show that the GATE is identifiable.

\textbf{Step 1: Identification of $R(\cdot)$}


By the law of iterated expectations, we have the following relationship:
\begin{eqnarray}
\label{eq: identification iterated expectation}
     &&\mathbb{E} \Big( Y^i(t) \,\Big|\, A_t^i = a, \{ S(j), f(j) \}_{j \leq t} \Big)   \\
    &=&  \mathbb{E}\Big[\mathbb{E}\Big[Y^i(t) \,\Big|\, A_t^i = a, \{ S(j),f(j) \}_{j \leq t}, \{  Y(j) \}_{j < t} \Big]\Big|
     A_t^i = a, \{ S(j),f(j) \}_{j \leq t}\Big].\nonumber
\end{eqnarray}

By taking out the inner expectation in Equation \eqref{eq: identification iterated expectation}, we can further derive:
\begin{align}
    & \mathbb{E}\Big[Y^i(t) \,\Big|\, A_t^i = a, \{ S(j),f(j) \}_{j \leq t}, \{  Y(j) \}_{j < t}\Big] \nonumber\\
    \stackrel{(i)}{=} & \mathbb{E}\Big[Y^i(t) \,\Big|\, A_t^i = a,\bar{\bm{A}}_{N,t-1}=\bar{\bm{a}}_{N,t-1}, \{ S(j),f(j) \}_{j \leq t}, \{  Y(j) \}_{j < t} \Big]\nonumber\\
    \stackrel{(ii)}{=} & \mathbb{E}\Big[Y^i(t, \bar{\bm{A}}_{N,t}) \,\Big|\, A_t^i = a,\bar{\bm{A}}_{N,t-1}=\bar{\bm{a}}_{N,t-1}, \{ S(j, \bar{\bm{A}}_{N,j-1}),f(j) \}_{j \leq t},\{  Y^i(j, \bar{\bm{A}}_{N,j}) \}_{j < t} \Big]\nonumber\\
    \stackrel{(iii)}{=} & \mathbb{E}\Big[Y^i(t, \bar{\bm{a}}_{N,t}) \,\Big|\, A_t^i = a,\bar{\bm{A}}_{N,t-1}=\bar{\bm{a}}_{N,t-1}, \{ S(j, \bar{\bm{A}}_{N,j-1}),f(j) \}_{j \leq t},\{  Y^i(j, \bar{\bm{A}}_{N,j}) \}_{j < t} \Big]\nonumber\\
    \stackrel{(iv)}{=} & \mathbb{E}\Big[Y^i(t, \bar{\bm{a}}_{N,t}) \,\Big|\, A_t^i = a, \{ S(j, \bar{\bm{a}}_{N,j-1}),f(j) \}_{j \leq t},\{  Y^i(j, \bar{\bm{a}}_{N,j}) \}_{j < t} \Big] \nonumber \\
    \stackrel{(v)}{=} & R\Big(a_t^i,\, \{ S(t), f(j) \}_{j \leq t} \Big) 
    \label{eq: identification conditional expectation}.
\end{align}
Here, equality (i) follows from Assumption \ref{assump: ATRA}, Assumption \ref{assump: SRA} and \ref{assump: PA}; equalities (ii)-(iv) follow from the consistency Assumption \ref{assump: CA}. 

According to the definition of $R(\cdot)$ and the independence of $\{  Y^i(j, \bar{\bm{a}}_{N,j}) \}_{j < t}$, substituting the result of Equation \eqref{eq: identification conditional expectation} into Equation \eqref{eq: identification iterated expectation}, and using the result of Assumption \ref{assump: SRA}, we obtain immediately that:
\begin{equation*}
    R\Big(a_t^i,\, \{ S(t), f(j) \}_{j \leq t} \Big)
    =\mathbb{E} \Big( Y^i(t) \,\Big|\, A_t^i=a, \{ S(j), f(j) \}_{j \leq t} \Big).
\end{equation*}


Consequently, $R(\cdot)$ can be estimated by the sample average over units receiving treatment $a$:
{\small
\begin{equation*}
    R\Big(a_t^i,\, \{ S(j), f(j) \}_{j \leq  t} \Big)
    =N_a^{-1}\mathbb{E} \Big(\sum_{\{i:a_t^i=a\}} Y^i(t) \,\Big|\, A_t^i=a_t^i, \{ S(j), f(j) \}_{j < t} \Big),
\end{equation*}
}
which completes the derivation of Equation \eqref{eq: definition of R}.

\textbf{Step 2: Identification of GATE}

We next show Equation \eqref{eq: result of idetification}. By the law of iterated expectations, we have
\begin{eqnarray*}
    &&\mathbb{E} \Big[ R \big( a, \{ S(j, \bar{\bm{a}}_{N,j-1} ), f(j) \} _{j \leq t}\big) \Big]\\
    &=&\mathbb{E} \Big[\mathbb{E} \Big\{ R \big( a, \{ S(j, \bar{\bm{a}}_{N,j-1} ), f(j) \} _{j \leq t} \big) \Big|\bar{A}_{N,1}=\bar{\bm a}_{N,1}, \{ S(j, \bar{\bm{a}}_{N,j-1} ), f(j) \} _{j \leq t}, \{  Y(j, \bar{\bm{a}}_{N,t} ) \} _{j < t} \Big\}\Big].
\end{eqnarray*}
Under Assumption \ref{assump: CA}, we can replace $Y(1,\bar{\bm a}_{N,1})$ and $S(2,\bar{\bm a}_{N,1})$ with $Y(1)$ and $S(2)$, respectively.  Under Assumption \ref{assump: SRA} and \ref{assump: PA}, the event $(A_2^1,\dots,A_2^N)=(a_2^1,\dots,a_2^N)$ can be included in the conditioning set. This yields that
\begin{align*}
    &\mathbb{E} \Big[  \big( a, \{ S(j, \bar{\bm{a}}_{N,j-1} ), f(j) \} _{j \leq t}\big) \Big]
    =\mathbb{E} \Big[\mathbb{E} \Big\{ R \big( a, \{ S(t, \bar{\bm{a}}_{N,j-1} ), f(j) \} _{j \leq t},  \big) \Big| \\
    &\bar{A}_{N,2}=\bar{\bm a}_{N,2}, \{ S(j, \bar{\bm{a}}_{N,j-1} ), f(j) \} _{j \leq t}, \{  Y(j, \bar{\bm{a}}_{N,t} ) \} _{j < t}, S(1), S(2), Y(1)\Big\}\Big].
\end{align*}

Iteratively applying this argument allows us to repeatedly replace the counterfactual variables with the
observed ones. At the end, all the potential outcomes/states will be replaced with the observed versions
conditional on the actions. The proof is hence completed.
\end{proof}

\subsection{Proof of Proposition \ref{prop: ATE expression}}\label{app:ATE expression}
\begin{proof}
Let $\zeta(t):=N^{-1}
       \mathbb{E} ( \sum_{i=1}^N  Y^i ( t,  \bm{1}_{N,t} ) )
       - N^{-1}
       \mathbb{E} ( \sum_{i=1}^N  Y^i ( t,  \bm{0}_{N,t} ) )$, which means that $\text{GATE}=\sum_{t=1}^m\zeta(t)$.
Thus, we only need to prove that
\begin{align*}
    \zeta(t)&= 
(  \alpha_{0,C_1} (t) -  \alpha_{0,C_0}(t) ) +
 \Big[ ( \alpha_{1,C_1} (t)  -  \alpha_{1,C_0} (t) )^\top\mathbb{E}( N^{-1} \sum_{i=1}^N X^i(t)  )   \Big] \\ 
&+\Big[
\alpha_{2,C_1} (t)^\top   \Big( \sum_{k=1}^{t-1} \{ \prod_{l=k+1}^{t-1} \Phi_{1,C_1}(l) [\gamma_{0, C_1}(k)+\gamma_{1, C_1}(k)+\Phi_{0, C_1} (k)\mathbb{E}( N^{-1} \sum_{i=1}^N H^i_{C_1}(k) )]  \}
+ \prod_{l=1}^{t-1} \Phi_{1,C_1}(l) \mathbb{E}(S(1)) \Big)
\\
&-
\alpha_{2,C_0} (t)^\top \Big(
\sum_{k=1}^{t-1} \{ \prod_{l=k+1}^{t-1} \Phi_{1, C_0}(l) [\gamma_{0, C_0}(k)+\Phi_{0, C_0} (k)\mathbb{E}( N^{-1} \sum_{i=1}^N H^i_{C_0}(k) )]  \}
+ \prod_{l=1}^{t-1} \Phi_{1, C_0}(l) \mathbb{E}(S(1)) \Big)
\Big].
\end{align*}

Since potential outcomes admits the form of models in \eqref{eq:model_Y}-\eqref{eq:model_state_within_group_all}, we have
\begin{align*}
    &\zeta(t)
    =(\alpha_{0,C_1}(t)-\alpha_{0,C_0}(t))+( \alpha_{1,C_1} (t)  -  \alpha_{1,C_0} (t) )^\top\mathbb{E}( N^{-1} \sum_{i=1}^N X^i(t) ) \\
    &+\alpha_{2,C_1}(t)^\top\mathbb{E} ( S_{C_1} ( t,  \bm{1}_{N,t-1} ) )
    - \alpha_{2,C_0}(t)^\top\mathbb{E} (S_{C_0} ( t,  \bm{0}_{N,t-1} ) )\\
    =&(\alpha_{0,C_1}(t)-\alpha_{0,C_0}(t))+( \alpha_{1,C_1} (t)  -  \alpha_{1,C_0} (t) )^\top\mathbb{E}( N^{-1} \sum_{i=1}^N X^i(t)  ) \\
    &+\alpha_{2,C_1}(t)^\top[\gamma_{0,C_1}(t-1)+\gamma_{1,C_1}(t-1)+\Phi_{0,C_1}(t-1)\mathbb{E} ( S_{C_1} ( t-1,  \bm{1}_{N,t-2} ) )+\Phi_{1,C_1}(t-1)\mathbb{E}( N^{-1} \sum_{i=1}^N H^i(t-1)  )]\\
    &- \alpha_{2,C_0}(t)^\top[\gamma_{0,C_0}(t-1)+\Phi_{0,C_0}(t-1)\mathbb{E} ( S_{C_0} ( t-1,  \bm{1}_{N,t-2} ) )+\Phi_{1,C_0}(t-1)\mathbb{E}( N^{-1} \sum_{i=1}^N H^i(t-1)  )]\\
    =&\cdots\\
    =& 
(  \alpha_{0,C_1} (t) -  \alpha_{0,C_0}(t) ) +
 \Big[ ( \alpha_{1,C_1} (t)  -  \alpha_{1,C_0} (t) )^\top\mathbb{E}( N^{-1} \sum_{i=1}^N X^i(t)  )   \Big] \\ 
&+\Big[
\alpha_{2,C_1} (t)^\top   \Big( \sum_{k=1}^{t-1} \{ \prod_{l=k+1}^{t-1} \Phi_{1,C_1}(l) [\gamma_{0, C_1}(k)+\gamma_{1, C_1}(k)+\Phi_{0, C_1} (k)\mathbb{E}( N^{-1} \sum_{i=1}^N H^i_{C_1}(k) )]  \}
+ \prod_{l=1}^{t-1} \Phi_{1,C_1}(l) \mathbb{E}(S(1)) \Big)
\\
&-
\alpha_{2,C_0} (t)^\top \Big(
\sum_{k=1}^{t-1} \{ \prod_{l=k+1}^{t-1} \Phi_{1, C_0}(l) [\gamma_{0, C_0}(k)+\Phi_{0, C_0} (k)\mathbb{E}( N^{-1} \sum_{i=1}^N H^i_{C_0}(k) )]  \}
+ \prod_{l=1}^{t-1} \Phi_{1, C_0}(l) \mathbb{E}(S(1)) \Big)
\Big].
\end{align*}

The proof is hence completed.
\end{proof}

\subsection{Proof of Theorem~\ref{thm:bootstrap_consistency_QTE}}\label{app:bootstrap_consistency_QTE}
\begin{proof}
We focus on provide an upper error bound for
\begin{equation*}
    \rho^*(z)=\left|\mathbb{P}\left(\frac{1}{m}\hat{T}-\frac{1}{m}T\leq z\right)-\mathbb{P}\left(\frac{1}{m}\hat{T}^b-\frac{1}{m}\hat{T}\leq z\bigg|Data\right)\right|.
\end{equation*}

We begin with some notations. Note that, in outcome model, $\widetilde{\theta}_{C}(t), C\in\{C_0,C_1\}$ can be expressed as
\begin{equation*}
    \widetilde{\theta}_{C}(t)=\theta_{s,C}(t)+\frac{1}{n}\sum_{d=1}^n\left(\sum_{j=1}^mB_{d,j,C}(t)e_{d,C}(j)\right),
\end{equation*}
where
\begin{equation*}
    B_{d,j,C}(t)=\omega_{j,h}(t)\left(\frac{1}{n}\sum_{d'=1}^nZ_{d',C}(j)^\top Z_{d',C}(j)\right)^{-1}Z_{d,C}(j),
\end{equation*}
are independent of the random part $e_{d,C}(j)$, and $\theta_{s,C}(t)=\sum_{j=1}^m\omega_{j,h}(t)\theta_C(j)$. Let $e_{d,C}^{\theta}(t)=\sum_{j=1}^mB_{d,j,C}(t)e_{d,C}(j)=\left\{e_{d,C}^{\alpha_{0,C}}(t),\left(e_{d,C}^{\alpha_{1,C}}(t)\right)^\top,\left(e_{d,C}^{\alpha_{2,C}}(t)\right)^\top\right\}^\top$ and $e_{C}^{\theta}(t)=n^{-1/2}\sum_{d=1}^ne_{d,C}^{\theta}(t)$. Similarly, we can represent $\widetilde{\Theta}(t)$ as
\begin{equation*}
    \widetilde{\Theta}(t)=\widetilde{\Theta}_{s,C}(t)+\frac{1}{n}\sum_{d=1}^n\left(\sum_{j=1}^{m-1}\widetilde{B}_{d,j,C}(t)E_{d,C}(t)\right),
\end{equation*}
where
\begin{equation*}
    \widetilde{B}_{d,j,C}(t)=\omega_{j,h}(t)\left(\frac{1}{n}\sum_{d'=1}^D\widetilde{Z}_{d',C}(j)^\top \widetilde{Z}_{d',C}(j)\right)^{-1}\widetilde{Z}_{d,C}(j),
\end{equation*}
are independent of the random part $E_{d,C}(j)$, and $\Theta_{s,C}(t)=\sum_{j=1}^m\omega_{j,h}(t)\Theta_C(j)$. Let $E_{d,C}^{\Theta}(t)=\sum_{j=1}^m\widetilde{B}_{d,j,C}(t)E_{d,C}(j)=\left\{E_{d,C}^{\gamma_{0,C}}(t),\left(E_{d,C}^{\Phi_{0,C}}(t)\right)^\top,\left(E_{d,C}^{\Phi_{1,C}}(t)\right)^\top\right\}^\top$ and $E_{C}^{\Theta}(t)=n^{-1/2}\sum_{d=1}^nE_{d,C}^{\Theta}(t)$. 

The OLS estimation corresponds to the special case $h=0$. We remark that $E_{C}^{\Theta}(t)$ is asymptotically normal when $h=0$ and degenerates to a point distribution when $mh\rightarrow\infty$. To make the following analysis hold for the OLS-based test statistic, we view $E_{C}^{\Theta}(t)$ as a random variable in the discussion below.

For simplicity, let $\text{vec}(\cdot)$ be the operator that reshapes a matrix into a vector by stacking its columns on top of one another. Denote
{\small
\begin{align*}
    &x_{d,C}(t)=\left[e_{d,C}^{\alpha_{0,C}}(t),\left(e_{d,C}^{\alpha_{1,C}}(t)\right)^\top,\left(e_{d,C}^{\alpha_{2,C}}(t)\right)^\top, \left(E_{d,C}^{\gamma_{0,C}}(t)\right)^\top,\left\{\text{vec}\left(E_{d,C}^{\Phi_{0,C}}(t)\right)\right\}^\top,\left\{\text{vec}\left(E_{d,C}^{\Phi_{1,C}}(t)\right)\right\}^\top\right]^\top, \\
    &x_{d}(t)=\left[x_{d,C_0}(t)^\top,x_{d,C_1}(t)^\top\right]^\top\in\mathbb{R}^{2+p+2q(q+p_H+2)},\\
    &x_{d}=\left(x_{d}(2)^\top,x_{d}(3)^\top,\dots,x_{d}(m)^\top\right)^\top\in\mathbb{R}^{p_x},\quad p_x=(m-1)[2+p+2q(q+p_H+2)],
\end{align*}
}
where $p_H$ is the dimension of $\widetilde{X}_{d,C}$ and $\widetilde{X}_{d,S}$.

Let $\{y_{d}\}_{d}$ be independent mean zero Gaussian vectors with $\mathbb{E}y_{d}y_{d}^\top=\mathbb{E}x_{d}x_{d}^\top$. We similarly represent $y_{d}$ as
{\small
\begin{align*}
    &y_{d,C}(t)=\left[\bar{e}_{d,C}^{\alpha_{0,C}}(t),\left(\bar{e}_{d,C}^{\alpha_{1,C}}(t)\right)^\top,\left(\bar{e}_{d,C}^{\alpha_{2,C}}(t)\right)^\top, \left(\bar{E}_{d,C}^{\gamma_{0,C}}(t)\right)^\top,\left\{\text{vec}\left(\bar{E}_{d,C}^{\Phi_{0,C}}(t)\right)\right\}^\top,\left\{\text{vec}\left(\bar{E}_{d,C}^{\Phi_{1,C}}(t)\right)\right\}^\top\right]^\top, \\
    &y_{d}(t)=\left[y_{d,C_0}(t)^\top,y_{d,C_1}(t)^\top\right]^\top\in\mathbb{R}^{2+p+2q(q+p_H+2)},\\
    &y_{d}=\left(y_{d}(2)^\top,y_{d}(3)^\top,\dots,y_{d}(m)^\top\right)^\top\in\mathbb{R}^{p_x},\quad p_x=(m-1)[2+p+2q(q+p_H+2)],
\end{align*}
}
Let $\{e_{d,C_0}^b(j),e_{d,C_1}^b(j),E_{d,C_0}^b(j),E_{d,C_1}^b(j)\}$ be the empirical Gaussian analogs of $\{e_{d,C_0}(j),e_{d,C_1}(j),E_{d,C_0}(j),E_{d,C_1}(j)\}$. In other words, for $d=1,\dots,n$, $j=1,\dots,m$, let
\begin{equation*}
    e_{d,C_0}^b(j)=\hat{e}_{d,C_0}(j)\xi_d,\;e_{d,C_1}^b(j)=\hat{e}_{d,C_1}(j)\xi_d,\;E_{d,C_0}^b(j)=\hat{E}_{d,C_0}(j)\xi_d,\;E_{d,C_1}^b(j)=\hat{E}_{d,C_1}(j)\xi_d,
\end{equation*}
where $\xi_1,\dots,\xi_n$ are i.i.d. standard normal random variables. We next define
{\small
\begin{align*}
    &w_{d,C}(t)=\left[\bar{e}_{d,C}^{\alpha_{0,C},b}(t),\left(\bar{e}_{d,C}^{\alpha_{1,C},b}(t)\right)^\top,\left(\bar{e}_{d,C}^{\alpha_{2,C},b}(t)\right)^\top, \left(\bar{E}_{d,C}^{\gamma_{0,C},b}(t)\right)^\top,\left\{\text{vec}\left(\bar{E}_{d,C}^{\Phi_{0,C},b}(t)\right)\right\}^\top,\left\{\text{vec}\left(\bar{E}_{d,C}^{\Phi_{1,C},b}(t)\right)\right\}^\top\right]^\top, \\
    &w_{d}(t)=\left[w_{d,C_0}(t)^\top,w_{d,C_1}(t)^\top\right]^\top\in\mathbb{R}^{2+p+2q(q+p_H+2)},\\
    &w_{d}=\left(w_{d}(2)^\top,w_{d}(3)^\top,\dots,w_{d}(m)^\top\right)^\top\in\mathbb{R}^{p_x},\quad p_x=(m-1)[2+p+2q(q+p_H+2)],
\end{align*}
}
Let
\begin{eqnarray*}
    &&X=(X_2^\top,X_3^\top,\dots,X_m^\top)=\frac{1}{\sqrt{n}}\sum_{d=1}^nx_{d},\\
    &&Y=(Y_2^\top,Y_3^\top,\dots,Y_m^\top)=\frac{1}{\sqrt{n}}\sum_{d=1}^ny_{d},\\
    &&W=(W_2^\top,W_3^\top,\dots,W_m^\top)=\frac{1}{\sqrt{n}}\sum_{d=1}^nw_{d}.
\end{eqnarray*}
Define the following function
\begin{eqnarray*}
    &&F_{\text{GATE}}(X;\theta_{C_0},\theta_{C_1},\Theta_{C_0},\Theta_{C_1})\\
    &\equiv&\frac{1}{m}\sum_{t=2}^m(  \alpha_{0,C_1} (t) -  \alpha_{0,C_0}(t) +\frac{e_{i,C}^{\alpha_{0,C_1}}(t)}{\sqrt{n}}-\frac{e_{i,C}^{\alpha_{0,C_0}}(t)}{\sqrt{n}})\\
    &&+\frac{1}{m}\sum_{t=1}^m \Big[ \mathbb{E}( n^{-1} \sum_{i=1}^n X^i(t)  ) ( \alpha_{1,C_1} (t)  -  \alpha_{1,C_0} (t) +\frac{e_{i,C}^{\alpha_{1,C_1}}(t)}{\sqrt{n}}-\frac{e_{i,C}^{\alpha_{1,C_0}}(t)}{\sqrt{n}})  \Big]\\
    &&+\frac{1}{m}\sum_{t=1}^m \Big[
\Big(\alpha_{2,C_1} (t)+\frac{e_{i,C}^{\alpha_{2,C_1}}(t)}{\sqrt{n}}\Big)^\top   \Big( \sum_{k=1}^{t-1} \{ \prod_{l=k+1}^{t-1} (\Phi_{2,C_1}(l)+\frac{E_{i,C}^{\Phi_{2,C_1}}(l)}{\sqrt{n}}) (\Phi_{0, C_1}(k)+\frac{E_{i,C}^{\Phi_{0,C_1}}(k)}{\sqrt{n}}\\
&&+
(\Phi_{1, C_1} (t)+\frac{E_{i,C}^{\Phi_{1,C_1}}(t)}{\sqrt{n}})\mathbb{E}( n^{-1} \sum_{i=1}^n H^i_{C_1}(t) )  )  \}+ \prod_{l=1}^{t-1} (\Phi_{2,C_1}(l)+\frac{E_{i,C}^{\Phi_{2,C_1}}(l)}{\sqrt{n}}) \mathbb{E}(S(1)) \Big)\\
&&-\Big(\alpha_{2,C_0} (t)+\frac{e_{i,C}^{\alpha_{2,C_0}}(t)}{\sqrt{n}}\Big)^\top   \Big( \sum_{k=1}^{t-1} \{ \prod_{l=k+1}^{t-1} (\Phi_{2,C_0}(l)+\frac{E_{i,C}^{\Phi_{2,C_0}}(l)}{\sqrt{n}}) (\Phi_{0, C_0}(k)+\frac{E_{i,C}^{\Phi_{0,C_0}}(k)}{\sqrt{D}}\\
&&+
(\Phi_{1, C_0} (t)+\frac{E_{i,C}^{\Phi_{1,C_0}}(t)}{\sqrt{n}})\mathbb{E}( n^{-1} \sum_{i=1}^n H^i_{C_0}(t) )  )  \}+ \prod_{l=1}^{t-1} (\Phi_{2,C_0}(l)+\frac{E_{i,C}^{\Phi_{2,C_0}}(l)}{\sqrt{D}}) \mathbb{E}(S(1)) 
\Big)\Big].
\end{eqnarray*}

We next represent the proposed test statistic and the bootstrap samples based on $F_{\text{GATE}}$. Recall that $\Theta_{s,C}(t)=\sum_{j=1}^m\omega_{j,h}(t)\Theta_C(j)$ and $\theta_{s,C}(t)=\sum_{j=1}^m\omega_{j,h}(t)\theta_C(j)$ are the smoothed parameters, and $\widetilde{\theta}_C, \widetilde{\Theta}_C$ correspond to the estimates. The difference between the proposed test statistic and the oracle indirect effect $m^{-1}(\hat{T}-\text{GATE})$ can be represented as $T_0^*=F_{\text{GATE}}(X;\theta_{s,C_0},\theta_{s,C_1},\Theta_{s,C_0},\Theta_{s,C_1})-F_{\text{GATE}}(0;\theta_{C_0},\theta_{C_1},\Theta_{C_0},\Theta_{C_1})$. Similarly, we can represent $T^{-1}(\widehat{\text{GATE}}^b-\widehat{\text{GATE}})$ by $W_0^*=F_{\text{GATE}}(W;\widetilde{\theta}_{C_0},\widetilde{\theta}_{C_1},\widetilde{\Theta}_{C_0},\widetilde{\Theta}_{C_1})-F_{\text{GATE}}(0;\widetilde{\theta}_{C_0},\widetilde{\theta}_{C_1},\widetilde{\Theta}_{C_0},\widetilde{\Theta}_{C_1})$. By definition, we have
\begin{equation*}
    \rho^*(z)=\Big|P\{T_0^*\leq z\}-P\{W_0^*\leq z\}\Big|.
\end{equation*}
We also define the oracle statistics: $T_0=F_{\text{GATE}}(X;\theta_{C_0},\theta_{C_1},\Theta_{C_0},\Theta_{C_1})-F_{\text{GATE}}(0;\theta_{C_0},\theta_{C_1},\Theta_{C_0},\Theta_{C_1})=F_{\text{GATE}}(X)-F_{\text{GATE}}(0)$, $W_0=F_{\text{GATE}}(W;\theta_{C_0},\theta_{C_1},\Theta_{C_0},\Theta_{C_1})-F_{\text{GATE}}(0;\theta_{C_0},\theta_{C_1},\Theta_{C_0},\Theta_{C_1})=F_{\text{GATE}}(W)-F_{\text{GATE}}(0)$ by replacing $\theta_{s,C}$, $\widetilde{\theta}_C$, $\Theta_{s,C}$ and $\widetilde{\Theta}_C$, $C\in\{C_0,C_1\}$, with the oracle values. This yields an upper bound for
\begin{equation*}
    \rho(z)=\Big|P\{T_0\leq z\}-P\{W_0\leq z\}\Big|.
\end{equation*}

Following the derivation of Lemma 4 in \citet{luo2024policy}, we know $\sup_z\rho(z)\leq \widetilde{C} n^{-1/8}$, for some constant $\widetilde{C}>0$.

Notice that
\begin{eqnarray*}
    \rho^*(z)&=&\Big|P\{T_0^*\leq z\}-P\{W_0^*\leq z\}\Big|\\
    &\leq&\Big|P\{T_0^*\leq z\}-P\{T_0\leq z\}\Big|+\Big|P\{W_0^*\leq z\}-P\{W_0\leq z\}\Big|+\Big|P\{T_0\leq z\}-P\{W_0\leq z\}\Big|\\
    &\leq&I_1+I_2+\widetilde{C} n^{-1/8},
\end{eqnarray*}
where $I_1$ and $I_2$ denote the above first two components, respectively. 

Define $T_{01}^*:=F_{\text{GATE}}(X;\theta_{s,C_0},\theta_{s,C_1},\Theta_{s,C_0},\Theta_{s,C_1})-F_{\text{GATE}}(0;\theta_{s,C_0},\theta_{s,C_1},\Theta_{s,C_0},\Theta_{s,C_1})$. Then $I_1$ can be further divided into two part:
\begin{eqnarray*}
    I_1&=&\Big|P\{T_0^*\leq z\}-P\{T_0\leq z\}\Big|\\
    &\leq&\Big|P\{T_0^*\leq z\}-P\{T_{01}^*\leq z\}\Big|+\Big|P\{T_{01}^*\leq z\}-P\{T_0\leq z\}\Big|\\
    &=&I_{11}+I_{12},
\end{eqnarray*}
where $I_{11}$ and $I_{12}$ denote the above two components, respectively. 

Similar to the proof of Theorem 2 in \citet{luo2024policy}, we can obtain that
\begin{equation*}
    I_2\leq \widetilde{C}n^{-1/8},\:
    I_{12}\leq \widetilde{C}n^{-1/8}.
\end{equation*}
Also, according to the proof of Theorem 2 in \citet{luo2024policy}, we know that
\begin{equation*}
    I_{11}=O(n^{1/2}h^2+n^{1/2}m^{-1}),
\end{equation*}
holds with probability $1$ as $n\rightarrow\infty$. The proof is hence completed.

\end{proof}

\end{document}